\documentclass[sigconf, nonacm]{acmart}

\newcommand{\nop}[1]{}

\newcommand\vldbdoi{XX.XX/XXX.XX}
\newcommand\vldbpages{XXX-XXX}
\newcommand\vldbvolume{14}
\newcommand\vldbissue{1}
\newcommand\vldbyear{2020}
\newcommand\vldbauthors{\authors}
\newcommand\vldbtitle{\shorttitle} 
\newcommand\vldbavailabilityurl{URL_TO_YOUR_ARTIFACTS}
\newcommand\vldbpagestyle{plain} 
\newtheorem{definition}{Definition}
\newtheorem{example}{Example}
\newtheorem{lemma}{Lemma}
\usepackage[linesnumbered,ruled,vlined]{algorithm2e}
\usepackage{amsthm}
\usepackage{bm}
\usepackage{float}
 \usepackage{balance}
\usepackage{hyperref}

\usepackage{enumitem}
\usepackage{tcolorbox}

\usepackage{subfigure}

\begin{document}

\title{S$^3$GND: An Effective Learning-Based Approach for Subgraph Similarity Search Under Generalized Neighbor Difference Semantics (Technical Report)}

\author{Qi Wen}
\affiliation{%
  \institution{East China Normal University}
  \city{Shanghai}
  \country{China}
}
\email{51265902057@stu.ecnu.edu.cn}

\author{Xiang Lian}
\affiliation{%
  \institution{Kent State University}
  \city{Kent}
  \state{Ohio}
  \country{USA}
}
\email{xlian@kent.edu}

\author{Nan Zhang}
\affiliation{%
  \institution{East China Normal University}
  \city{Shanghai}
  \country{China}
}
\email{52275902026@stu.ecnu.edu.cn}

\author{Yutong Ye}
\affiliation{%
  \institution{East China Normal University}
  \city{Shanghai}
  \country{China}
}
\email{52205902007@stu.ecnu.edu.cn}

\author{Mingsong Chen}
\affiliation{%
  \institution{East China Normal University}
  \city{Shanghai}
  \country{China}
}
\email{mschen@sei.ecnu.edu.cn}

\begin{abstract}

Subgraph similarity search over large-scale graphs is a fundamental task that retrieves subgraphs similar to a given query graph from a data graph, and it plays a crucial role in real applications such as protein discovery, social network analysis, and recommendation systems. While prior works on subgraph similarity search studied various graph similarity metrics, in this paper, we propose a novel graph similarity semantics, \textit{generalized neighbor difference} (GND), that accounts for both the keyword-set relationships between vertices and edge-weight differences. We formulate the problem of \textit{subgraph similarity search under the generalized neighbor difference semantics} (S$^3$GND), which retrieves those subgraphs similar to a query graph $q$ under GND semantics. To efficiently tackle the S$^3$GND problem, we propose an effective learning-based approach, which constructs a keyword hypergraph from the data graph, and trains a \textit{hypergraph neural network} (HGNN) model to obtain high-quality keyword embedding representations. We design effective pruning strategies, \textit{keyword embedding MBR}, \textit{vertex-Level ND lower bound}, and \textit{graph-level GND lower bound pruning}, to rule out false alarms of candidate vertices/subgraphs, and devise a tree-based indexing mechanism to facilitate efficient S$^3$GND query answering. We develop an efficient S$^3$GND query-processing algorithm that traverses the index, applies pruning strategies, and returns actual S$^3$GND answers. Finally, we conduct extensive experiments to verify the effectiveness and efficiency of our proposed S$^3$GND approach over both real and synthetic graphs.

\end{abstract}

\maketitle

\pagestyle{\vldbpagestyle}
\begingroup\small\noindent\raggedright\textbf{PVLDB Reference Format:}\\
\vldbauthors. \vldbtitle. PVLDB, \vldbvolume(\vldbissue): \vldbpages, \vldbyear.\\
\href{https://doi.org/\vldbdoi}{doi:\vldbdoi}
\endgroup
\begingroup
\renewcommand\thefootnote{}\footnote{\noindent
This work is licensed under the Creative Commons BY-NC-ND 4.0 International License. Visit \url{https://creativecommons.org/licenses/by-nc-nd/4.0/} to view a copy of this license. For any use beyond those covered by this license, obtain permission by emailing \href{mailto:info@vldb.org}{info@vldb.org}. Copyright is held by the owner/author(s). Publication rights licensed to the VLDB Endowment. \\
\raggedright Proceedings of the VLDB Endowment, Vol. \vldbvolume, No. \vldbissue\ %
ISSN 2150-8097. \\
\href{https://doi.org/\vldbdoi}{doi:\vldbdoi} \\
}\addtocounter{footnote}{-1}\endgroup

\ifdefempty{\vldbavailabilityurl}{}{
\vspace{.3cm}
\begingroup\small\noindent\raggedright\textbf{PVLDB Artifact Availability:}\\
The source code, data, and/or other artifacts are available at \url{https://github.com/Luminous-wq}.
\endgroup
}


\begin{figure}[!t]
    \centering
    \subfigure[collaboration network $G$]{
        \scalebox{0.26}[0.26]{\includegraphics{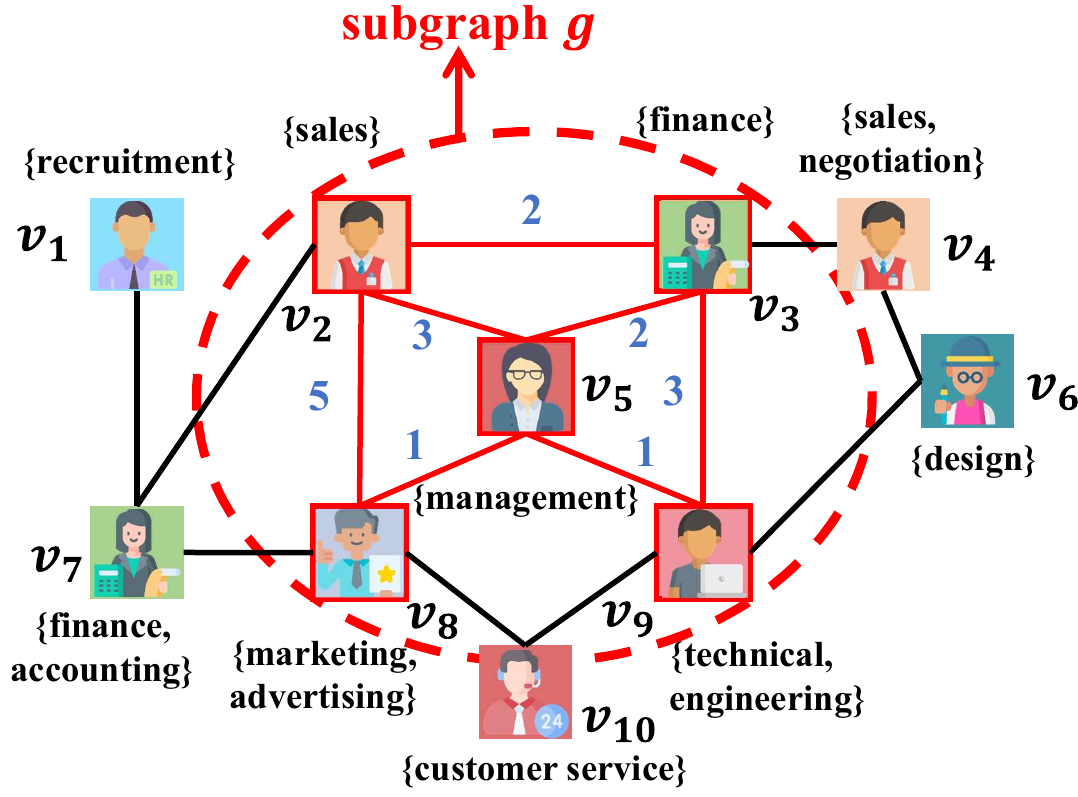}}
        \label{fig:exa1}
    }
    \subfigure[target team $q$]{
        \raisebox{4ex}{\scalebox{0.26}[0.26]{\includegraphics{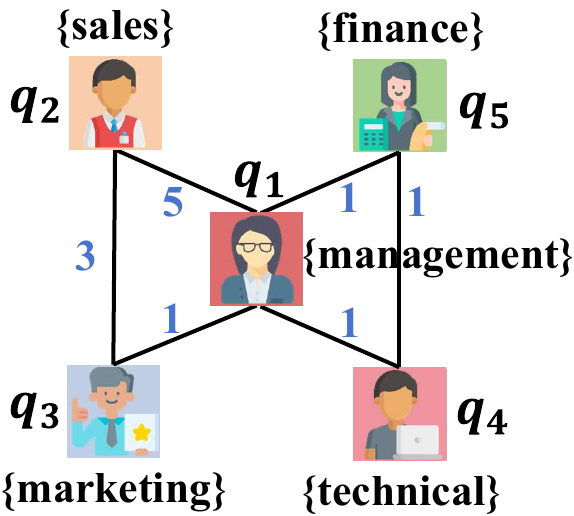}}}
        \label{fig:exa2}
    }\vspace{-3ex}
    \caption{An S$^3$GND example of the skilled team search.}\vspace{-3ex}
    \label{fig:problem_example}
        \vspace{-0.1in}
\end{figure}

\section{Introduction}
\label{sec:intro}
\enlargethispage{2\baselineskip}

One of the fundamental tasks in graph data management is the \textit{subgraph similarity search} problem, which retrieves subgraphs $g$ from a data graph $G$ that are similar to a given query graph $q$. Efficient and effective subgraph similarity search has gained increasing attention from the database community, and has been widely used in many important fields such as chemistry~\cite{shiokawa2025gesim, yoshimori2025context}, bioinformatics~\cite{shen2014applied, bonnici2024arcmatch}, and social network analysis~\cite{heimann2018regal, hassine2023non}.

In this paper, we focus on a large-scale, weighted data graph model, where vertices are associated with keyword sets and edges with weights. We propose a novel metric, namely \textit{generalized neighbor difference} (GND), to measure the similarity between two \textit{weighted} (sub)graphs by accounting for both the keyword-set relationships between matching vertices and the weight differences between edges (from vertices to their neighbors). Under GND semantics, we formulate an important problem, named \textit{subgraph similarity search under the generalized neighbor difference semantics} (S$^3$GND), which retrieves subgraphs $g$ from a weighted graph $G$ that match with a given query graph $q$ (i.e., with GND scores below a threshold).


We have the following motivation example of our S$^3$GND problem in the real application of project team recommendation.

\begin{example}\textbf{(The Skilled Team Search from Weighted Collaboration Networks)} 
\label{example:S3GND_example}
Figure~\ref{fig:exa1} shows a weighted collaboration network $G$ with 10 employee vertices, $v_1 \sim v_{10}$, each of whom has a set of skill keywords (e.g., employee $v_4$ is skilled at ``sales’’ and ``negotiation’’). Each edge in network $G$ is associated with a weight, indicating the past collaboration count (frequency) of two employees (e.g., employees $v_2$ and $v_3$ have collaborated twice before, thus, edge $e_{v_2, v_3}$ has a weight $2$).


A manager plans to launch a new project and seek a desired team (e.g., target team, $q$, as shown in Figure~\ref{fig:exa2}) with the required skills and collaborative experience. Specifically, each team member $q_i$ ($1 \leq i \leq 5$) must have the desired skills (e.g., member $q_3$ is required to possess the ``marketing’’ skill). Moreover, members with specific skills must share a strong collaborative relationship, as evidenced by past collaborations, to reduce communication costs for the project. In the example of Figure~\ref{fig:exa2}, the ``marketing'' member $q_3$ is expected to have previously collaborated with the ``sales'' member $q_2$ at least 3 times (i.e., the weight of edge $e_{q_3, q_2}$ is $3$).


In this case, the manager can conduct an S$^3$GND query to identify some team (e.g., the subgraph $g$ within the dashed circle in Figure~\ref{fig:exa1}) in the collaboration network $G$, such that (1) each team member in subgraph $g$ has all the required skills specified by the corresponding query vertex in $q$ (e.g., $v_8$ matches $q_3$, as $v_8$'s keyword set \{marketing, advertising\} contains $q_3$'s \{marketing\}), and; (2) the past collaboration frequencies of members (i.e., edge weights) should meet the weight constraints of query edges as much as possible (e.g., the collaboration count 5 between $v_2$ and $v_8$ in $g$ is above the threshold 3 between $q_3$ and $q_2$ in $q$). Note that, if an edge weight (i.e., past collaboration frequency) cannot meet the weight threshold of a query edge (e.g., the edge weight 3 between $v_2$ and $v_5$ in $g$ is lower than the threshold 5 between $q_2$ and $q_1$ in $q$), it may incur extra communication cost (relevant to the GND score). Thus, the S$^3$GND will return teams (e.g., subgraph $g$) with such GND scores as low as possible.  \qquad $\blacksquare$


\end{example}


In the example above, the S$^3$GND query can find subgraphs (i.e., teams) in a large-scale, weighted data graph (i.e., a collaboration network) that are similar to the query graph pattern under GND semantics (with strict keyword set containment relationships between vertices and small edge-weight distribution distances).


Note that prior work on the subgraph similarity search proposed various graph similarity metrics, either on an unweighted data graph with each vertex labeled by a single keyword~\cite{sanfeliu2012distance,blumenthal2020comparing,dutta2017neighbor,li2018efficient}, or over an unweighted data graph containing vertices associated with keyword sets~\cite{wen2025s3and}. In contrast, our GND semantic is designed for a more general data graph model, that is, a \textit{weighted data graph} containing vertices with keyword sets (instead of a single keyword per vertex) and edges associated with weights (rather than unweighted edges), which is not uncommon in many real-world applications such as bibliographical networks (i.e., author nodes with research interests as keywords and collaboration edges with the number of co-author publications as weights), social networks (i.e., user nodes with profile attribute values as keywords and friendship edges with influences as weights), and so on. Thus, we cannot directly apply prior techniques to solve our S$^3$GND problem.



Answering S$^3$GND queries over a large-scale data graph is challenging due to the enormous number of candidate subgraphs and the computationally expensive evaluation of keyword set and GND score constraints. To efficiently tackle the S$^3$GND problem, we design a novel learning-based approach that uses a \textit{hypergraph neural network} (HGNN) to obtain high-quality keyword embeddings. We also devise effective pruning strategies for keyword embedding MBRs and vertex- / graph-level GND scores, and construct a tree-based indexing mechanism over offline pre-computed data (i.e., keyword embedding MBRs and edge weight distributions) to facilitate efficient online retrieval of S$^3$GND answers.


Specifically, we made the following contributions in this paper.

\begin{enumerate}
    \item We formulate a novel problem, named \textit{subgraph similarity search under the generalized neighbor difference semantics} (S$^3$GND) over a large-scale, weighted data graph containing a keyword set per vertex in Section~\ref{sec:problem_definition}.
    
    \item We propose a two-phase framework to efficiently tackle the S$^3$GND problem in Section~\ref{sec:framework}.
    
    \item We design effective pruning strategies in Section~\ref{sec:Pruning} to reduce the S$^3$GND search space, enabling which we train a \textit{hypergraph neural network} (HGNN) to compute keyword embeddings and derive the edge-weight distribution distances (i.e., GND scores).
    
    \item We devise offline pre-computation and a tree-based indexing mechanism to facilitate efficient S$^3$GND query answering in Section~\ref{sec:offline}.
    
    \item We develop an online S$^3$GND query-processing algorithm that traverses the tree index and applies the pruning strategies proposed in Section~\ref{sec:online}.
    
    \item We demonstrate through extensive experiments the efficiency and effectiveness of our S$^3$GND query processing approach over real/synthetic graphs in Section~\ref{sec:experiment}.
\end{enumerate}

Section~\ref{sec:related_work} overviews the related work on (learning-based) subgraph similarity search. Finally, Section~\ref{sec:conclusion} concludes this paper.

\section{Problem Definition}
\label{sec:problem_definition}
In this section, we provide the graph data model, present the \textit{generalized neighbor difference} (GND) semantics for subgraph matching, and formally define the problem of \textit{subgraph similarity search under the generalized neighbor difference semantics} (S$^3$GND).

\subsection{Graph Data Model}

In this subsection, we first define the graph data model.

\begin{definition}
    \textbf{(Data Graph, \bm{$G$})} A data graph $G$ is in the form of a triple $(V(G), E(G), \Gamma(G))$, where $V(G)$ is a set of vertices, $v_i$, in graph $G$, each with a keyword set $v_i.K$, $E(G)$ represents a set of edges $e_{v_i, v_j}$ (connecting two ending vertices $v_i$ and $v_j$), each associated with a positive weight $e_{v_i, v_j}.w$, and $\Gamma(G)$ is a mapping function: $V(G) \times V(G) \rightarrow E(G)$.
    \label{def:data_graph}
\end{definition}

The graph data model (given in Definition \ref{def:data_graph}) can be used in many real applications such as social networks (containing user vertices with profile attributes and friendship edges with influence intensities), bibliographical networks (containing author vertices with certain research expertise and edges that indicate the collaboration relationships, with the number of co-authored papers as edge weights), collaboration networks, and so on.

\subsection{Generalized Neighbor Difference (GND) Semantics}
In this subsection, we define the \textit{neighbor difference} (ND) and the \textit{generalized neighbor difference} (GND) semantics.

\noindent{\textbf{The Vertex-to-Vertex Mapping, $M: V(q) \rightarrow V(g)$:} We consider the  matching between a query graph $q$ and a subgraph $g$ with the same graph size. If there is a vertex-to-vertex mapping, $M: V(q) \rightarrow V(g)$, between $q$ and $g$, then it holds that $q_j.K \subseteq v_i.K$, for each matching pair of query vertex $q_j \in V(q)$ and data vertex $v_i \in V(g)$ (denoted as $M(q_j)=v_i$).


\begin{definition}
    \label{def:ND}
    \textbf{(Neighbor Difference, \bm{$ND(q_j, v_i)$})} 
    Given a target vertex, $q_j$, of a query graph $q$, a vertex, $v_i$, of a subgraph $g$, and a vertex-to-vertex mapping $M$ (i.e., $M(q_j)=v_i$), their \textit{neighbor difference} (ND), $ND(q_j, v_i)$, is given by:
    \begin{equation}
         ND(q_j, v_i) = \hspace{-2ex}\sum_{\forall n_j \in N(q_j), n_i \in N(v_i), M(n_j) = n_i} \hspace{-3ex} \max\big\{e_{q_j, n_j}.w 
         -e_{v_i,n_i}.w, 0\big\},
        \label{eq:ND}
    \end{equation}
    where $N(\cdot)$ is a set of 1-hop neighbor vertices.
\end{definition}

Note that, prior works on S$^3$AND semantics \cite{wen2025s3and} considered the \textit{edge existence} between a vertex $q_j$ (or $v_i$) and its neighbors in $N(q_j)$ (or $N(v_i)$) for \textit{unweighted graphs}. Instead, the neighbor difference (ND) in Definition~\ref{def:ND} takes into account \textit{weighted subgraphs} with the intensity of edges (i.e., edge weights $e_{v_i, n_i}.w$) that can reach that of target (query) edges (i.e., edge weights $e_{q_j, n_j}.w$) in the query graph $q$. Thus, the neighbor difference in the S$^3$AND semantics \cite{wen2025s3and} is a special case of our ND definition, where edge weights are 1's if edges exist (0's, otherwise).

Next, we define the generalized neighbor difference (GND), which aggregates neighbor differences across matching vertices.

\begin{definition}
    \label{def:GND}
    \textbf{(Generalized Neighbor Difference, \bm{$GND(q,g)$})} Given a query graph $q$, a subgraph $g$, and a vertex-to-vertex mapping $M$ from $V(q)$ to $V(g)$ (note: $|V(q)|=|V(g)|$), the \textit{generalized neighbor difference} (GND), $GND(q,g)$, between $q$ and $g$ is defined as the aggregation over neighbor differences of all the matching vertex pairs $(q_j, v_i)$, that is,
    \begin{equation}
        GND(q,g)=f\left(\left\{ ND(q_j, v_i) | \forall (q_j, v_i), s.t. M(q_j) = v_i \right\}\right),
        \label{eq:GND}
    \end{equation}
    where $ND(q_j, v_i)$ is given by Eq.~(\ref{eq:ND}), and $f(S)$ is an aggregate function (e.g., MAX, AVG, or SUM) over a set $S$.
\end{definition}


In Definition \ref{def:GND}, function $f(\cdot)$ (given by Eq.~(\ref{eq:GND})) aggregates the NDs for all vertex pairs via MAX/AVG/SUM predicate, which outputs the maximum, average, or summed vertex neighbor differences between query graph $q$ and subgraph $g$, respectively. Intuitively, the GND score measures the dissimilarity between $q$ and $g$ with respect to graph structure and edge weights. 

In Example \ref{example:S3GND_example}, the GND score determines the maximum/average /overall past collaboration intensity and, in turn, the familiarity with other team members in the team collaboration network. For AVG and SUM aggregates (i.e., $f(\cdot)=AVG/SUM$), GND scores imply the amortized performance for the team collaboration, whereas for the MAX aggregate (i.e., $f(\cdot)=MAX$), the GND score can identify the weakest vertex matching (with the most dissimilar neighbor connections/edges). In this paper, we primarily focus on $MAX$ and $SUM$ aggregates. The AVG aggregate can be obtained by dividing $SUM$ by the query graph size $|V(q)|$ (constant). Thus, we will reserve the interesting topic of considering the AVG aggregation function for our future work.



\subsection{The S$^3$GND Problem}

In this subsection, we formally define the problem of \textit{subgraph similarity search under the generalized neighbor difference semantics} (S$^3$GND), which retrieves subgraphs $g$ from a data graph $G$, matching with a given query graph $q$ with low GND score.

\begin{definition}
    \label{def:S3GND}
    \textbf{(Subgraph Similarity Search Under Generalized Neighbor Difference Semantics, \bm{$S^3GND(G,q)$})} Given a data graph $G$, a query graph $q$, a vertex-to-vertex mapping $M: V(q)\to V(g)$, and a generalized neighbor difference threshold $\delta$, a \textit{subgraph similarity search under the generalized neighbor difference semantics} (S$^3$GND) retrieves connected subgraphs $g$ of $G$, such that:
    \begin{itemize}
        \item {\bf (Equal Subgraph Size)} $|V(q)|=|V(g)|$;
        \item {\bf (Keyword Set Containment)} for the mapping vertices $q_j\in V(q)$ and $v_i\in V(g)$, it holds that $q_j.K\subseteq v_i.K$, and;
        \item {\bf (Generalized Neighbor Difference)} the generalized neighbor difference satisfies the condition that $GND(q,g) \leq \delta$,
    \end{itemize}
    where $GND(q,g)$ is given by Eq.~(\ref{eq:GND}).
\end{definition}


Intuitively, S$^3$GND returns subgraphs $g$ similar to the query graph $q$ of the same size (i.e., $|V(q)|=|V(g)|$). Here, each query vertex $q_j\in V(q)$ matches with a similar vertex $v_i\in V(g)$, if their keyword sets satisfy the containment relationship (i.e., $q_j.K \subseteq v_i.K$), and their GND score is low (i.e., $GND(q,g) \leq \delta$).

As shown in Example~\ref{example:S3GND_example}, we can issue an S$^3$GND query to find a team that meets the skill requirements and has extensive collaboration experience, thereby significantly reducing the onboarding costs for the new team. Such a team is more likely to succeed based on the 
query graph pattern customized for the project.


\vspace{0.5ex}\noindent {\bf Challenges:}
Efficient and effective answering of S$^3$GND queries is rather challenging, due to the enormous S$^3$GND search space. First, it is very costly to enumerate many subgraphs (of the same size as query graph $q$) that satisfy the keyword set constraints in a large-scale data graph $G$. Second, calculating GND scores for weighted subgraph candidates is computationally expensive. It is therefore highly desirable to design effective approaches to efficiently and scalably retrieve S$^3$GND query answers.


\begin{table}[t]
\caption{\small Symbols and Descriptions}
\vspace{-2ex}
\scriptsize
\label{tab1}
\begin{center}
\begin{tabular}{|l|p{6.5cm}|}
\hline
\textbf{Symbol}&{\textbf{Description}} \\
\hline\hline
$G$ & a data graph\\
\hline
$V(G)$ & a set of vertices $v_i$\\
\hline
$E(G)$ & a set of edges $e_{v_i,v_j}$\\
\hline
$q$ & a query graph\\
\hline
$g$ & a subgraph of data graph $G$\\
\hline
$v_i.K$ & a keyword set associated with vertex $v_i \in V(G)$\\
\hline
$e_{v_i,v_j}.w$ & a positive weight of edge $e_{v_i, v_j} \in E(G)$\\
\hline
$N(v_i)$ & a set of vertex $v_i$'s 1-hop neighbors \\
\hline
$ND(q_j, v_i)$ & the 1-hop neighbor difference between vertices $q_j$ and $v_i$ w.r.t. edge weights \\
\hline
$GND(q, g)$ & the generalized neighbor difference (GND) score between query graph $q$ and subgraph $g$  \\
\hline
$\delta$ & a threshold for the GND score \\
\hline
$\Phi$ & a trained keyword embedding model\\
\hline
$o(\cdot)$ & the output embedding vector of $\Phi$\\
\hline
$H$ & a keyword hypergraph\\
\hline
\end{tabular}
\end{center}\vspace{-3ex}
\end{table}

Table \ref{tab1} depicts the commonly used notations and their descriptions in this paper.

\section{The S$^3$GND Query Framework}
\label{sec:framework}



Algorithm~\ref{alg:S3GND} outlines the pseudo code of our S$^3$GND query processing framework, which aims to efficiently search subgraphs similar to a given query graph under GND semantics, consisting of two phases: \textit{offline pre-computation} (lines 1-5) and \textit{online S$^3$GND query processing phases} (lines 6-8).
Specifically, in the \textit{offline pre-computation} phase, we offline train a keyword embedding model, $\Phi$, to generate embedding vectors for vertex keywords, such that the size of \textit{Minimum Bounding Rectangle} (MBR)~\cite{papadias1997spatial} over all keyword embedding vectors in each vertex is as small as possible (line 1). Then, for each vertex $v_i$ in the data graph $G$, we use the trained keyword embedding model $\Phi$ to compute the keyword embedding vector $o(v_i.k_j)$ for each keyword $v_i.k_j\in v_i.K$ (lines 2-3). Next, we use the keyword embedding vector to compute the \textit{keyword embedding MBR}, $v_i.MBR$, bounding all keyword embedding vectors $o(v_i.k_j)$ in each vertex $v_i$, and compute a sorted list of edge weights, $v_i.list_w$ (lines 4-5). Finally, we store the pre-computed data $v_i.MBR$ and $v_i.list_w$ in an auxiliary data structure, $v_i.Aux$ (line 6), over which
we construct a tree index $\mathcal{I}$ to facilitate online retrieval, pruning, and querying (line 7).

\setlength{\textfloatsep}{0pt}
\begin{algorithm}[t]
    \caption{\bf  S$^3$GND Query Processing Framework}
      \small
    \label{alg:S3GND}
    \KwIn{
        \romannumeral1) a data graph $G$, 
        \romannumeral2) a query graph $q$, and
        \romannumeral3) a generalized neighbor difference threshold $\delta$
    }
    \KwOut{
        a set, $S$, of subgraphs $g$ matching with the query graph $q$ under GND semantics
    }
    \tcp{\bf Offline Pre-Computation Phase}

    train a keyword embedding model $\Phi$ to compute an embedding vector, $o(v_i.k_j)$, for each keyword $v_i.k_j\in v_i.K$ in vertices $v_i$

    \For{each vertex $v_i \in V(G)$}{
        
        use model $\Phi$ to obtain a keyword embedding vector $o(v_i.k_j)$ for each keyword $v_i.k_j\in v_i.K$
        


        obtain the keyword embedding MBR, $v_i.MBR$

        compute a sorted list of edge weights, $v_i.list_w$

        store $(v_i.MBR,v_i.list_w)$ in an auxiliary data structure, $v_i.Aux$

    }

    construct a tree index $\mathcal{I}$ with pre-computed auxiliary data for the graph $G$
    
    \tcp{\bf Online S$^3$GND Query Processing Phase}

    \For{each S$^3$GND qeury}{

        traverse $\mathcal{I}$ by applying pruning strategies (i.e., keyword set pruning, vertex-level ND lower bound pruning, and graph-level GND lower bound pruning) to retrieve candidate vertices $q_j.V_{cand}$ of each query vertex $q_j\in V(q)$

        assemble/refine candidate subgraphs $g$ and return a set, $S$, of actual S$^3$GND subgraph answers
    
    } 
\end{algorithm}
\setlength{\textfloatsep}{0pt}

In the \textit{online S$^3$GND query processing} phase, for each S$^3$GND query, we traverse the index $\mathcal{I}$ by applying the pruning strategies (i.e., keyword embedding MBR pruning, vertex-level ND lower bound pruning, and graph-level GND lower bound pruning, as discussed later in Section \ref{sec:Pruning}), and retrieve candidate vertices in $q_j.V_{cand}$ for each query vertex $q_j \in V(q)$ (lines 8-9).
Finally, we join candidate vertices of query vertices in $q_j.V_{cand}$ into candidate subgraphs $g$, and refine them to return a set, $S$, of actual subgraph answers that meet S$^3$GND predicates (line 10).

Figure \ref{fig:framework} illustrates a workflow of the S$^3$GND query processing framework.

\begin{figure}[!t]
    \centering
    \includegraphics[width=0.9\linewidth]{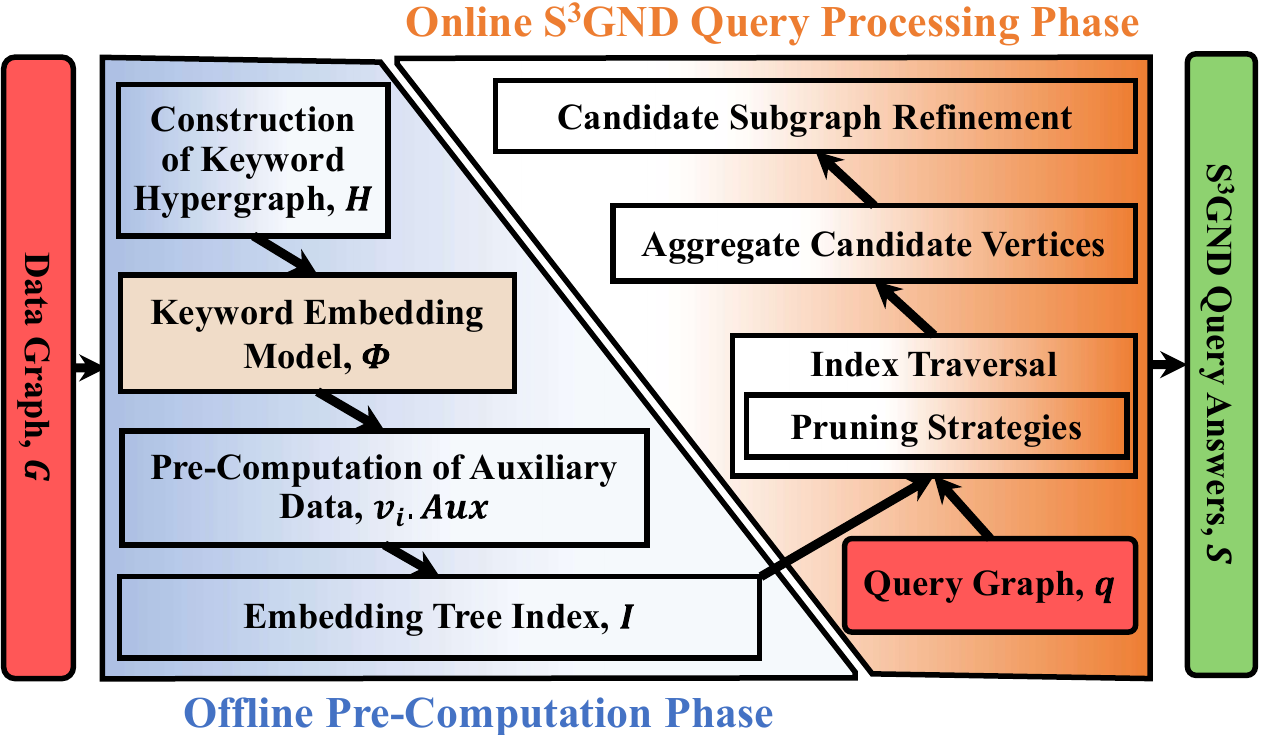}
     \vspace{-0.1in}
    \caption{Workflow for processing S$^3$GND queries.}
    \label{fig:framework}\vspace{2ex}
\end{figure}

\section{Pruning Strategies}
\label{sec:Pruning}
In this section, we present effective pruning strategies (i.e., keyword embedding MBR pruning, vertex-level ND lower bound pruning, and graph-level GND lower bound pruning) that reduce the search space of the S$^3$GND problem.

\subsection{Keyword Embedding MBR Pruning}

In Definition~\ref{def:S3GND}, the keyword set of each candidate vertex $v_i$ in candidate subgraph $g$ must contain all the query keywords in the matching query vertex $q_j$ in query graph $q$. Our goal is to design a learning model to obtain keyword embeddings, and use a \textit{minimum bounding rectangle} (MBR), $v_i.MBR$, to minimally bound keyword embeddings in each vertex $v_i$. Then, we can devise an effective \textit{keyword embedding MBR pruning} method to filter out candidate vertices that do not meet the keyword set containment relationship, as given below.

\begin{lemma}
    \label{lemma:keyword_pruning}
    \textbf{(Keyword Embedding MBR Pruning)} Given a candidate vertex $v_i$, its keyword embedding MBR $v_i.MBR$, and a query graph $q$, vertex $v_i$ can be safely pruned, if it holds that: $v_i.MBR \cap q_j.MBR \neq q_j.MBR$ (i.e., $q_j.MBR  \not\subseteq v_i.MBR$), for all query vertices $q_j \in V(q)$.
\end{lemma}
\begin{proof} If two sets of keywords have a subset relationship (e.g., $q_j.K \subseteq v_i.K$), their corresponding sets of keyword embedding vectors also satisfy the containment relationship. Based on the MBR property, the keyword embedding MBR $q_j.MBR$ of $q_j.K$ must be inside the other one, $v_i.MBR$, of $v_i.K$ (i.e., $q_j.MBR  \subseteq v_i.MBR$).

Therefore,  if $v_i.MBR \cap q_j.MBR \neq q_j.MBR$ holds for a candidate vertex $v_i$ for all query vertices $q_j \in V(q)$, then query keyword sets $q_j.K$ cannot be subsets of $v_i.K$. According to the keyword set containment property in Definition~\ref{def:S3GND}, we can infer that vertex $v_i$ cannot match with any of the query vertices $q_j$ in the query graph $q$. Hence, we can safely prune the candidate vertex $v_i$, completing the proof. 
\end{proof}

From Lemma \ref{lemma:keyword_pruning}, we can safely prune those vertices $v_i$ whose keyword embedding MBRs $v_i.MBR$ does not fully contain that, $q_j.MBR$, of any query vertex $q_j$. In other words, these vertices $v_i$ cannot match with any of the query vertices in the query graph $q$ due to the keyword set constraint, and thus can be ruled out.


\subsection{Vertex-Level ND Lower Bound Pruning}

We propose a \textit{vertex-level ND lower bound pruning} strategy to filter out those vertices that have high neighbor differences and thus cannot match with query vertices in the query graph $q$.


\begin{lemma}
    \label{lemma:ND_lb_pruning}
    \textbf{(Vertex-Level ND Lower Bound Pruning)} Given a candidate vertex $v_i$, a query vertex $q_j$, and an generalized neighbor difference threshold $\delta$, candidate vertex $v_i$ can be safely pruned w.r.t. $q_j$, if it holds that $lb\_ND(q_j,v_i) > \delta$, where $lb\_ND(q_j,v_i)$ is a lower bound of $ND(q_j,v_i)$.
\end{lemma}
\begin{proof}
Since $lb\_ND(q_j,v_i)$ is a lower bound of $ND(q_j,v_i)$, we have $ND(q_j,v_i) \geq lb\_ND(q_j,v_i)$. Due to the lemma assumption that $lb\_ND(q_j,v_i) > \delta$, via the inequality transition, it holds that $ND(q_j,v_i) \geq lb\_ND(q_j,v_i)> \delta$.

Based on Definition~\ref{def:GND}, for the aggregate function $f(\cdot)$ like MAX or SUM, we have $GND(q,g) \geq ND(q_j,v_i)$. By applying the inequality transition again, we obtain $GND(q,g) \geq ND(q_j,v_i) > \delta$. Thus, from Definition~\ref{def:S3GND}, we know that vertex $v_i$ in any candidate subgraph $g$ cannot match with query vertex $q_j$ in the query graph, due to a high GND score (i.e., violating the GND condition that $GND(q,g)\leq \delta$). Hence, $v_i$ can be safely pruned with respect to $q_j$.
\end{proof}

In Lemma \ref{lemma:ND_lb_pruning}, candidate vertex $v_i$ cannot match with query vertex $q_j$, if its ND lower bound $lb\_ND(q_j,v_i)$ is greater than threshold $\delta$. Intuitively, if the ND lower bound is greater than $\delta$,  their ND score (in turn GND score with MAX or SUM aggregate function) is definitely greater than $\delta$, which violates the GND condition in Definition \ref{def:S3GND}. Thus, we can use such a vertex-level ND lower bound to enable the pruning.


\subsection{Graph-Level GND Lower Bound Pruning}

In this subsection, we present a \textit{graph-level GND lower bound pruning} method to filter out candidate subgraphs with high GND scores (i.e., $GND(q, g) > \delta$). 


\begin{lemma}
    \label{lemma:GND_Lower_Bound_Pruning}
    \textbf{(Graph-Level GND Lower Bound Pruning)} Given a candidate subgraph $g$, a query graph $q$, and a generalized neighbor difference threshold $\delta$, subgraph $g$ can be safely pruned, if it holds that $lb\_GND(q,g) > \delta$, where $lb\_GND(q,g)$ is a lower bound of $GND(q,g)$.
\end{lemma}

\begin{proof}
    Since $lb\_GND(q,g)$ is a lower bound of the generalized neighbor difference $GND(q,g)$, we have $GND(q,g) \geq lb\_GND(q,g)$. From the lemma assumption that $lb\_GND(q,g) > \delta$ holds, by the inequality transition, it holds that $GND(q,g) \geq lb\_GND(q,g) > \delta$. Thus, from Definition~\ref{def:S3GND}, candidate subgraph $g$ cannot be the S$^3$GND answer and can be safely pruned, which completes the proof.
\end{proof}

Intuitively, the graph-level GND lower bound pruning (as given in Lemma \ref{lemma:GND_Lower_Bound_Pruning}) filters out those subgraphs $g$ with GND score lower bounds (in turn, GND scores) greater than threshold $\delta$, which introduces no false negatives.

\vspace{0.5ex}\noindent\textbf{Discussions on How to Compute Keyword Embedding MBRs and ND/GND Score Lower Bounds:} In order to enable our pruning strategies, we need to offline pre-calculate keyword embedding MBRs for keyword sets of vertices, and sorted lists of edge weights for ND/GND score lower bounds. In this paper, we propose an effective learning-based approach to compute high-quality keyword embeddings via a well-designed \textit{keyword embedding model} over a keyword hypergraph, which will be discussed later in Section \ref{subsec:keyword_embedding_MBR}. Moreover, we will also pre-compute a sorted list of edge weights for each vertex in graph $G$, which can be used for deriving lower bounds of ND/GND scores (as will be discussed later in Section \ref{subsec:ND_GND_LB}).

\section{Offline Pre-Computations}
\label{sec:offline}
In this section, we discuss how to pre-process the data graph offline to facilitate efficient online S$^3$GND queries, by pre-computing auxiliary data (including keyword embedding MBRs and edge weight lists) and constructing a tree index, $\mathcal{I}$, over such an auxiliary data structure.


\setlength{\textfloatsep}{0pt}

\subsection{Computing Keyword Embedding MBRs}
\label{subsec:keyword_embedding_MBR}

In order to support the keyword embedding MBR pruning (as given in Lemma \ref{lemma:keyword_pruning}), we propose an effective learning-based approach to compute high-quality keyword embeddings. In particular, we first construct a keyword hypergraph $H$ from the original data graph $G$, then train a keyword embedding model $\Phi$, and finally use model $\Phi$ to compute keyword embedding vectors (in turn MBRs).


\noindent{\bf Keyword Hypergraph Construction:} We aim to learn good keyword embeddings that can achieve high pruning power of keyword embedding MBR pruning. Intuitively, we would like to obtain keyword embedding MBRs as small as possible. In order to compute good keyword embeddings, we will transform the original data graph $G$ into a \textit{keyword hypergraph}, $H$, on which a learning model $\Phi$ can be trained and high-quality keyword embeddings can be obtained. Note that the keyword hypergraph $H$ usually consumes much smaller space than the original data graph (as confirmed by our empirical study on the space cost analysis in Figure~\ref{fig:space}), and the model training can thus be much more efficient.


We first define a keyword hypergraph $H$ below.
\begin{definition}
    \label{def:hypergraph}
    \textbf{(Keyword Hypergraph, $H$)} A keyword hypergraph is represented by a triple ($V(H)$, $E(H)$, $\chi$), where $V(H)$ denotes a set of keyword vertices, $k_i$, $E(H)$ is a set of hyperedges $e_{k_i,...,k_j}$, each connecting two or more vertices $k_i$, ..., and $k_j$ and associated with a weight $w_{k_i,...,k_j}$, and $\chi$ is a $|V(H)| \times |E(H)|$ incidence matrix of hypergraph $H$.
\end{definition}

}

    In Definition~\ref{def:hypergraph}, each hyperedge $e_{k_i,...,k_j}$ is associated with a weight $w_{k_i, ..., k_j}$, which is defined as the frequency of the keyword set (or the keyword combination) $\{k_i, ..., k_j\}$ that appears in vertices of the original data graph $G$. Moreover, $\chi$ denotes the incidence matrix of a hypergraph $H$, each element, $\chi[k_i] [e_{k_j,...,k_z}]$, of which indicates the incidence relationship between the keyword vertex $k_i$ and hyperedge $e_{k_j,...,k_z}$ ($i\in \{j,...,z\}$), defined as:
    \begin{equation}
        \chi[k_i] [e_{k_j,...,k_z}] = 
        \begin{cases}
            1, & \text{if } i\in \{j,...,z\};\\
            0, & \text{otherwise}.
        \end{cases}
        \label{eq:chi}
    \end{equation}


\begin{figure}[!t]
    \centering
    \includegraphics[width=0.95\linewidth]{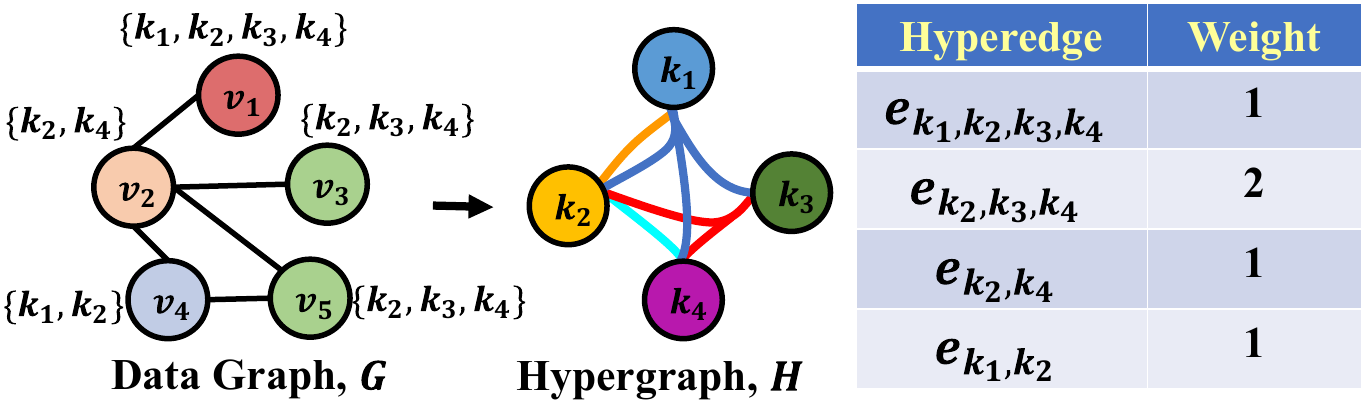}
    \vspace{-0.1in}
    \caption{An example of a keyword hypergraph construction.}
    \label{fig:hypergraph}\vspace{-2ex}
\end{figure}

\begin{example}
(\textbf{Keyword Hypergraph Construction})
Figure~\ref{fig:hypergraph} illustrates an example of the keyword hypergraph construction process. 
Specifically, each unique keyword in the data graph $G$ corresponds to a keyword vertex $k_i$ ($1\leq i\leq 4$) in the keyword hypergraph $H$, whereas the frequency of a keyword set (or a keyword combination) becomes the weight of the corresponding hyperedge.


As shown in Figure~\ref{fig:hypergraph}, data graph $G$ contains keywords $\{k_1,k_2,k_3,k_4\}$ (i.e., the keyword domain), which are transformed to four keyword vertices ``$k_1$'', ``$k_2$'', ``$k_3$'', and ``$k_4$'' in hypergraph $H$. The keyword set, $\{k_1,k_2,k_3,k_4\}$, in vertex $v_1$ of data graph $G$ is converted into a hyperedge $e_{k_1, k_2, k_3, k_4}$ connecting four keyword vertices, associated with weight (frequency) 1. Similarly, vertices $v_3$ and $v_5$ of graph $G$ have the same keyword set $\{k_2,k_3,k_4\}$. Thus, the transformed hyperedge $e_{k_2, k_3, k_4}$ in keyword hypergraph $H$ has the weight (frequency) 2. \qquad $\blacksquare$



\end{example}

 \begin{figure}[t]
    \centering
    \includegraphics[width=0.95\linewidth]{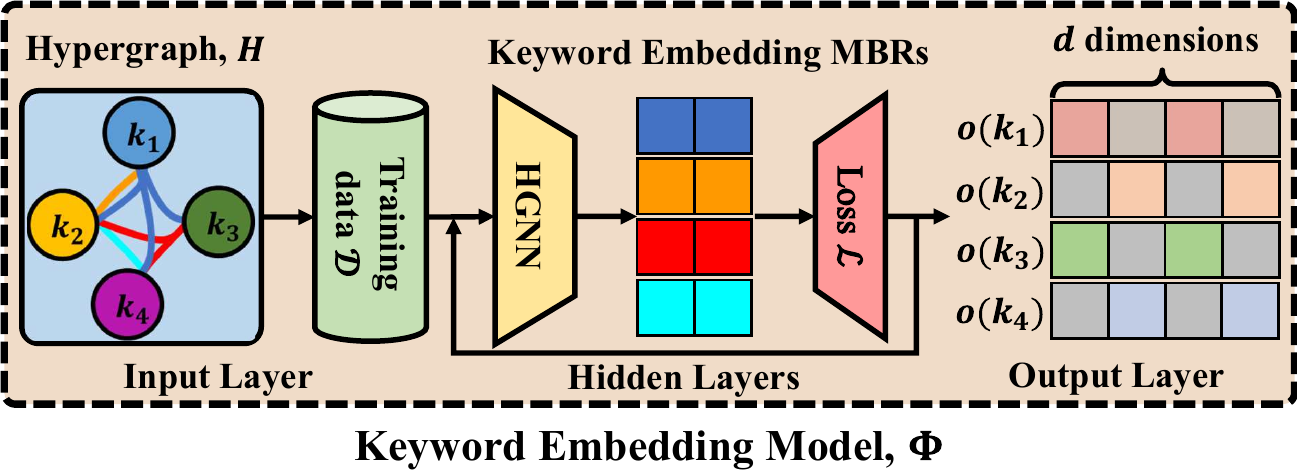}
     \vspace{-0.1in}
    \caption{A keyword embedding model, $\Phi$, for computing high-quality keyword embeddings.} 
    \label{fig:embedding}\vspace{-3ex}
\end{figure}

\noindent{\bf The Learning Framework for a Keyword Embedding Model:} Next, we illustrate how to train/learn a keyword embedding model, $\Phi$, over the keyword hypergraph $H$ in Figure~\ref{fig:embedding}, which can obtain keyword embeddings of high quality (or high pruning power) for our online S$^3$GND search.

Specifically, we train a \textit{hypergraph neural network} (HGNN)~\cite{feng2019hypergraph} over hypergraph $H$, via a carefully-designed contrastive loss w.r.t. keyword embedding MBRs, and a hyperedge sampling mechanism. 




\underline{\it The Structure of the Keyword Embedding Model:} The structure of HGNN over hypergraph $H$ consists of $L$ layers to learn a $d$-dimensional embedding vector, $o(k_i)$, for each keyword vertex $k_i$. In particular, for $0\leq l\leq L$, the $(l+1)$-th layer $O^{(l+1)}$  ($\in \mathbb{R}^{|V(H)| \times d}$) aggregates messages (embeddings) of node neighbors from the $l$-th layer $O^l$ ($\in \mathbb{R}^{|V(H)| \times d}$). That is, we have:
\begin{equation}
    O^{(l+1)} = \sigma(D^{-1/2}_v \chi \overline{W} D_e^{-1}\chi^T D_v^{-1/2} O^{(l)}\Theta^{(l)}),
    \label{eq:HGNN}
\end{equation}
where $O^{(0)} \in \mathbb{R}^{|V(H)| \times d_{in}}$ ($l=0$) is an initial feature matrix (e.g., keyword-based random initialization), for the input feature dimension $d_{in}$.

Note that, in Eq.~(\ref{eq:HGNN}), $\sigma(\cdot)$ is an activation function (e.g., $ReLU$), the incidence matrix $\chi$ is given by Eq.~(\ref{eq:chi}), $\overline{W}\in \mathbb{R}^{|E(H)| \times |E(H)|}$ is the normalized hyperedge weight matrix $W$, $\Theta^{(l)} \in \mathbb{R}^{d_{l} \times d_{l+1}}$ is a learnable parameter matrix (for input and output feature dimensions, $d_l$ and $d_{l+1}$, respectively, on the $l$-th layer), and $D_v \in \mathbb{R}^{|V(H)| \times |V(H)|}$ and $D_e \in \mathbb{R}^{|E(H)| \times |E(H)|}$ are diagonal matrices of vertex and hyperedge degrees, respectively. Here, each diagonal element in $D_v$ represents the sum of the weights of all hyperedges connected to keyword vertex $k_i$ (i.e., $D_v[k_i][k_i] = \sum_{e\in E(H)} \chi[k_i][e] \cdot \overline{W}[e][e]$); each diagonal element in $D_e$ is the number of keyword vertices contained in the hyperedge $e_{k_i, \cdots,k_j}$ (i.e., $D_e[e_{k_i, \cdots,k_j}][e_{k_i, \cdots,k_j}] = \sum_{z=1}^{|V(H)|}\chi[k_z][e_{k_i, \cdots,k_j}]$).








\underline{\it The Contrastive Loss Function, $\mathcal{L}$:} We carefully design a \textit{contrastive loss function} $\mathcal{L}$, with which we aim to minimize the number of false positives w.r.t. keyword set containment constraint in S$^3$GND predicates. In other words, if a query keyword set $q_j.K$ is not a subset of a vertex keyword set $v_i.K$, then we would like to minimize the chance that the query keyword embedding MBR, $q_j.MBR$, is inside the keyword embedding MBR, $v_i.MBR$.

When we train the learning model $\Phi$, we do not know the query keyword embedding MBR in advance. Thus, we will instead train the model $\Phi$ offline, using pairwise keyword embedding MBRs, $e_{k_i,\cdots,k_j}.MBR$ and $e_{k_z,\cdots,k_n}.MBR$, for hyperedges $e_{k_i,\cdots,k_j}$ and $e_{k_z,\cdots,k_n}$, respectively, in keyword hypergraph $H$.

The following are three cases for the relationship between two keyword sets $A=\{k_i,\cdots,k_j\}$ and $B=\{k_z,\cdots,k_n\}$, upon which we define the contrastive loss function $\mathcal{L}$.

 \begin{itemize}
     \item 
    {\bf (Case 1: Keyword Set Containment)} When the containment relationship $A \subseteq B$ holds, we want to minimize the intersection area of their MBRs, that is, $area(e_A.MBR \cap e_B.MBR) = area (e_A.MBR)$,  where $area(\cdot)$ is the area/volume of the MBR (implying the MBR size). We let the loss of this hyperedge pair $(e_A, e_B)$ be: $$\mathcal{L}_1 (e_A, e_B) = area(e_A.MBR).$$

     \item 
    {\bf (Case 2: Keyword Set Intersection)} When two keyword sets $A$ and $B$ have an intersection (but not containment) relationship, that is, $A \cap B \neq \emptyset$ (but $A \nsubseteq B$), we also want to minimize the intersection area of their MBRs. Thus, we set the loss to: $$\mathcal{L}_2 (e_A, e_B) = area(e_A.MBR \cap e_B.MBR).$$

     \item 
    {\bf (Case 3: Keyword Set Non-Intersection)} When two keyword sets $A$ and $B$ do not have any overlaps, that is, $A \cap B = \emptyset$, there are two possible sub-cases below, regarding the relationships of their keyword embedding MBRs.
    \begin{itemize}
        
        \item {\bf (Case 3.1: Intersecting MBRs)} If $e_A.MBR \cap e_B.MBR \neq \emptyset$ holds, we set the loss of hyperedge pair $(e_A, e_B)$ to: $$\qquad\qquad\mathcal{L}_{3.1} (e_A, e_B) = area(e_A.MBR \cap e_B.MBR).$$
        
        \item {\bf (Case 3.2: Non-Overlapping MBRs)} If it holds that $e_A.MBR \cap e_B.MBR = \emptyset$, we aim to minimize the areas of both MBRs to reduce the risk of false positives. Thus, we set the loss to: $$\qquad\qquad\mathcal{L}_{3.2} (e_A, e_B) = area(e_A.MBR) + area(e_B.MBR).$$
        
    \end{itemize}
    
        
        
        
 \end{itemize}

Any hyperedge pair $(e_A, e_B)$ falls into one of the cases (or sub-cases) above. We can therefore obtain our training data set, $\mathcal{D}$, containing hyperedge pairs $(e_A, e_B)$ in the corresponding categories $\mathcal{D}_1$, $\mathcal{D}_2$, $\mathcal{D}_{3.1}$, and $\mathcal{D}_{3.2}$, respectively. 

Now we are ready to define our contrastive loss function, $\mathcal{L}$, for our keyword embedding model $\Phi$ as follows.
\begin{equation}
    \label{eq:loss}
    \mathcal{L} = \sum_{\forall \mathcal{D}_x\in \{\mathcal{D}_1, \mathcal{D}_2, \mathcal{D}_{3.1}, \mathcal{D}_{3.2}\}}\mathbb{E}_{(e_A, e_B)\in \mathcal{D}_x} \mathcal{L}_x(e_A, e_B).
\end{equation}
where $\mathbb{E}(\cdot)$ denotes the expectation operator. 

We aim to train our keyword embedding model $\Phi$ over hypergraph $H$ that minimizes the loss function $\mathcal{L}$ (given in Eq.~(\ref{eq:loss})). Intuitively, the loss function $\mathcal{L}$ minimizes the size(s) of keyword embedding MBR(s) (or their intersections), so that the chance of false positives (i.e., $q_j.MBR$ is inside an MBR $e_B.MBR$, but $q_j.K \nsubseteq B$) is low.

\underline{\it Implementation:} Note that, in our implementation, in order to avoid the \textit{gradient vanishing} problem \cite{tan2019vanishing, hanin2018neural}, we replace MBR areas in the loss function $\mathcal{L}$ (as given in Eq.~(\ref{eq:loss})) with the logarithm of MBR areas. For example, for the area, $area(e_A.MBR)$, of hyperedge $e_A$'s keyword embedding MBR, $e_A.MBR$, we use:
\begin{eqnarray}
  area(e_A.MBR)&=&log\left(\prod_{i=1}^{d} \left(e_A.MBR[2\cdot i+1]-e_A.MBR[2\cdot i]\right)\right)\notag\\
            &=& \sum_{i=1}^{d}log\left(e_A.MBR[2\cdot i+1]-e_A.MBR[2\cdot i]\right).\notag
\end{eqnarray}

The logarithm, $log(\cdot)$, is taken as well in other MBR area calculations of our loss function $\mathcal{L}$.




\underline{\it The Hyperedge Sampling for the Model Training:} It is not efficient to train the keyword embedding model $\Phi$, by enumerating all pairs of hyperedges $(e_A, e_B)$ as the training data set $\mathcal{D}$ over a large-scale hypergraph $H$, which incurs quadratic cost $O(|E(H)|^2)$. Therefore, unlike traditional vertex sampling~\cite{jha2015path,bai2021efficient}, in this paper, we propose to use hyperedge sampling from hypergraph $H$ (rather than all hyperedges in $E(H)$), in order to enable efficient and scalable model training. Specifically, we construct the training set $\mathcal{D}$ by sampling hyperedge subsets $\mathcal{D}_1$, $\mathcal{D}_2$, $\mathcal{D}_{3.1}$, and $\mathcal{D}_{3.2}$ with an initial size ratio of $2:2:1:1$.

\begin{algorithm}[!t]
    \caption{\bf The Training Process of the Keyword Embedding Model}
    \label{alg:train}
      \small
    \KwIn{
        \romannumeral1) a data graph, $G$, and
        \romannumeral2) \# of training iterations, $iters$
    }
    \KwOut{a trained keyword embedding model, $\Phi$}

    \tcp{training data preparation phase}
    
    transform graph $G$ to keyword hypergraph $H$

    sample hyperedge pairs in hypergraph $H$ and obtain training data set $\mathcal{D}=\mathcal{D}_1 \cup \mathcal{D}_2 \cup \mathcal{D}_{3.1} \cup \mathcal{D}_{3.2}$

    initialize an HGNN model $\Phi$ for hypergraph $H$

    \tcp{hypergraph training phase}
    
    \For{epoch = 1 to $iters$}{
        
        $O \leftarrow \Phi(H)$

        compute $e_{k_i, ..., k_j}.MBR$ of each hyperedge $e_{k_i, ..., k_j}$ in $\mathcal{D}$
        
        obtain the contrastive loss function $\mathcal{L}$ via Eq.~(\ref{eq:loss})

        $\Phi \leftarrow optimize(\Phi, \mathcal{L})$ \qquad\textit{// backpropagation}
        
    }
    \Return $\Phi$

\end{algorithm}
\setlength{\textfloatsep}{0pt}


\underline{\it The Training Details:} The detailed training process is outlined in Algorithm~\ref{alg:train}. Specifically, we first prepare a training data set $\mathcal{D}$, by converting the data graph $G$ into a keyword hypergraph $H$, and sampling hyperedge pairs in $H$ as training data (lines 1-2). Next, we initialize an HGNN model $\Phi$ over hypergraph $H$ for the training (line 3). During the hypergraph training phase, for each iteration, we first obtain the current keyword vertex embeddings in $O$ from model $\Phi$, then compute keyword embedding MBR $e_{k_i, ..., k_j}.MBR$ of each hyperedge $e_{k_i, ..., k_j}$ in the training data set, and finally obtain the contrastive loss function $\mathcal{L}$ (given by Eq.~(\ref{eq:loss})) (lines 4-7). With the loss $\mathcal{L}$, we perform the backpropagation and update the parameters of the model $\Phi$ (line 8). Once the training is complete, we return the trained keyword embedding model $\Phi$ (line 9).




\subsection{Computing ND and GND Lower Bounds}
\label{subsec:ND_GND_LB}

In order to enable the pruning strategies in Lemmas \ref{lemma:ND_lb_pruning} and \ref{lemma:GND_Lower_Bound_Pruning}, we propose efficient approaches to compute an ND lower bound $lb\_ND(q_j, v_i)$ and a GND lower bound $lb\_GND(q,g)$, respectively, based on the summed difference between sorted lists of neighboring edge weights.

\noindent{\bf The Computation of ND Lower Bound, $lb\_ND(q_j, v_i)$:} Given a query vertex $q_j$ and a candidate vertex $v_i$, we compute a lower bound, $lb\_ND(q_j, v_i)$, of the ND score $ND(q_j, v_i)$ as follows.

First, we obtain two sorted lists, $q_j.list_w$ and $v_i.list_w$, of 1-hop neighbor edge weights for $q_j$ and $v_i$, respectively, in non-ascending order. That is, for $q_z \in N(q_j)$ and $v_z \in N(v_i)$, we have: $$q_j.list_w [z]= e_{q_j, q_z}.w,$$ and $$v_i.list_w [z] = e_{v_i, v_z}.w,$$ where $q_j.list_w [z]\geq q_j.list_w [z+1]$ and $v_i.list_w [z]\geq v_i.list_w [z+1]$. Note that, when edge $e_{q_j, q_z}$ (or $e_{v_i, v_z}$) does not exist, we have edge weight $e_{q_j, q_z}.w = 0$ (or $e_{v_i, v_z}.w =0$).

Next, we have the following lemma to compute an ND lower bound, $lb\_ND(q_j, v_i)$.

\begin{lemma} \textbf{(ND Lower Bound Computation)}
Given a query vertex $q_j$ and a candidate vertex $v_i$, the lower bound, $lb\_ND(q_j, v_i)$, of the ND score $ND(q_j, v_i)$ is given by:
\begin{equation}
    \label{eq:lb_ND}
    lb\_ND(q_j, v_i) = \sum_{z=1}^{|N(q_j)|} \max\{q_j.list_w[z]-v_i.list_w[z], 0\}.
\end{equation}
\label{lemma:lb_ND}
\end{lemma}
\begin{proof}

    Let $R=q_j.list_w$ and $T=v_i.list_w$ be two sorted lists of length $n=|N(q_j)|$.
    We first consider a function $P(r, t)=\max\{r-t,0\}$. 
    Then, the ND score $ND(q_j, v_i)$  between $q_j$ and $v_i$ (as given in Eq.~(\ref{eq:ND})) can be rewritten as: 
    \begin{eqnarray}
      ND(q_j, v_i)&=&\sum_{\forall n_j \in N(q_j), n_i \in N(v_i), M(n_j) = n_i} \hspace{-3ex} \max\big\{e_{q_j, n_j}.w 
         -e_{v_i,n_i}.w, 0\big\}\notag\\
         &=&    \sum_{\forall R[z]= n_j \in N(q_j), n_i \in N(v_i), M(n_j) = n_i} P(R[z], e_{v_i,n_i}.w).  \notag
    \end{eqnarray}
    \noindent where $e_{v_i, n_i}.w \in T$.

    \underline{\it Property of $P(r, t)$:} We first give the property of $P(r, t)$. Specifically, for any $r_1 \geq r_2$ in $R$ and $t_1 \geq  t_2$ in $T$, function $P(r,t)$ has the property that: $$P(r_1,t_1)+P(r_2,t_2)\leq P(r_1,t_2)+P(r_2,t_1).$$
    
    Note that, $P(r, t)$ can be viewed as the length from $t$ to $r$ on the number line. Specifically, if $r>t$, the length is $r-t$, otherwise (i.e., $r\leq t$), the length is 0.
    Therefore, we can express $P(r, t)$ in integral form, i.e., $$P(r, t)=\int_{0}^{\infty} 1_{\{t < x \leq r\}} \, dx.$$
    
    Then, the property inequality is equivalent to $$\int_{0}^{\infty}(1_{x\in(t_1,r_1]}+1_{x\in(t_2,r_2]})dx \leq \int_{0}^{\infty}(1_{x\in(t_2,r_1]}+1_{x\in(t_1,r_2]})dx.$$ 
    Since $r_1\geq r_2$ and $t_1\geq t_2$, we observe that the interval $(t_2,r_1]$ encompasses both $(t_1,r_1]$ and $(t_2,r_2]$. This implies that the combination of intervals on the right side of the inequality covers the combination of intervals on the left side in terms of the set. Therefore, the inequality $P(r_1,t_1)+P(r_2,t_2)\leq P(r_1,t_2)+P(r_2,t_1)$ holds when the conditions $r_1\geq r_2$ and $t_1\geq t_2$ are satisfied.

    
    

    


    \underline{\it Derivation of ND Lower Bound:} Since the list $R = q_j.list_w$ is sorted in non-ascending order, we always have $R[z]\geq R[z+1]$. If $R[z]$ and $R[z+1]$ have the corresponding weights $t_1$ and $t_2$ in $T$ (via vertex mapping $M(\cdot)$), respectively, and $t_1 \leq t_2$ holds, then we can always swap $t_1$ with $t_2$ to obtain a lower bound of ND score $ND(q_j, v_i)$. Overall, there exist a finite number of such swaps (e.g., bubble sort with $O(n\cdot log (n))$ cost) to sort weights of edges $e_{v_i,n_i}$ in subgraph $g$, such that the resulting ND score formula always becomes lower (or remains the same) for each swapping. After edge weights w.r.t. vertex $v_i$ have been sorted in non-ascending order, resulting in the sorted list $T = v_i.list_w$, we can now obtain a lower bound, $lb\_ND(q_j, v_i)$, of ND score $ND(q_j, v_i)$ below:
    \begin{eqnarray}
        lb\_ND(q_j, v_i) &=& \sum_{z=1}^{|N(q_j)|} P(R[z], T[z]) \notag\\
        &=&\sum_{z=1}^{|N(q_j)|}  \max\{q_j.list_w[z]-v_i.list_w[z], 0\},\notag
    \end{eqnarray}
    which completes the proof.
    
\end{proof}




\noindent{\bf The Computation of GND Lower Bound, $lb\_GND(q,g)$:} Similarly, we calculate the lower bound, $lb\_GND(q,g)$, of the GND score $GND(q,g)$ on the graph level below.



\begin{lemma}  \textbf{(GND Lower Bound Computation)}
Given a query graph $q$, a candidate subgraph $g$, and a vertex-to-vertex mapping $M$ from $V(q)$ to $V(g)$, the lower bound, $lb\_GND(q, g)$, of the GND score $GND(q, g)$ is given by:
\begin{equation}
    \label{eq:lb_GND}
    lb\_GND(q, g) = f(lb\_ND(q_j, v_i) | \forall (q_j, v_i), s.t. M(q_j) = v_i),
\end{equation}
where $f(\cdot)$ is a $MAX$ or $SUM$ aggregation function.
\label{lemma:lb_GND}
\end{lemma}
\begin{proof}
    According to Definition~\ref{def:GND}, the GND score is calculated by the aggregation function $f(\cdot)$ from the ND scores of each vertex pair $(q_j, v_i)$. From Lemma~\ref{lemma:lb_ND}, for each pair of matching vertices $(q_j, v_i)$, the ND score is always greater than or equal to its theoretical lower bound (i.e., $ND(q_j, v_i) \geq lb\_ND(q_j, v_i)$).
    Since the aggregation function $f(\cdot)$ (i.e., $MAX$ or $SUM$) is monotonically non-decreasing in all dimensions, we have: $f(ND(q_j, v_i) | \forall (q_j, v_i), s.t. M(q_j) = v_i)\geq f(lb\_ND(q_j, v_i) | \forall (q_j, v_i), s.t. M(q_j) = v_i)$. Substituting into Eq.~(\ref{eq:GND}), we have: 
    \begin{eqnarray}
        GND(q,g)&\geq& f(lb\_ND(q_j, v_i) | \forall (q_j, v_i), s.t. M(q_j) = v_i) \notag\\
        &=& lb\_GND(q, g).\notag
    \end{eqnarray}

    Hence, Eq.~(\ref{eq:lb_GND}) provides a lower bound of the GND score $GND(q, g)$, which completes the proof.
\end{proof}

\begin{algorithm}[!t]
    \caption{\bf Pre-Computations of Auxiliary Data}
    \label{alg:offline}
      \small
    \KwIn{
        a data graph $G$
    }
    \KwOut{
        the auxiliary data $v_i.Aux$ for each vertex $v_i$
    }

    \tcp{\bf keyword embedding training}

    construct a keyword hypergraph $H$ and train a keyword embedding model $\Phi$

    \For{each vertex $v_i \in V(G)$}{

        \tcp{\bf vertex keyword embedding MBR}

        use model $\Phi$ to obtain the keyword embedding vector $o(v_i.k_j)$ for each keyword $v_i.k_j \in v_i.K$
        
        use the keyword embedding vectors to compute the vertex keyword embedding MBR $v_i.MBR$

        \tcp{\bf edge weights with neighbors}

        calculate a sorted list, $v_i.list_w$, of edge weights for the neighborhood of vertex $v_i$
        
        $v_i.Aux = (v_i.MBR, v_i.list_w)$
    }
    
    \Return $v_i.Aux$
\end{algorithm}

\subsection{Indexing Mechanism}
\label{sec:index}

In this subsection, we first present how to offline pre-compute auxiliary data, and then discuss the construction of a tree index $\mathcal{I}$ over such auxiliary data to facilitate efficient online S$^3$GND querying.

\noindent{\bf Pre-Computation of Auxiliary Data, $v_i.Aux$:} As mentioned in Sections \ref{subsec:keyword_embedding_MBR} and \ref{subsec:ND_GND_LB}, in order to enable the pruning with keyword embedding MBRs and ND/GND lower bounds, respectively, we offline pre-compute the following auxiliary data, $v_i.Aux$, for each vertex $v_i$:
\begin{itemize}
    \item a keyword embedding MBR, $v_i.MBR$, minimally bounding embedding vectors $o(k_i)$ of all keywords $k_i \in v_i.K$, via the trained keyword embedding model $\Phi$, and;

    \item an ordered list, $v_i.list_w$, of edge weights, which is computed by sorting edge weights, $e_{v_i,n_i}.w$, between vertex $v_i$ and its 1-hop neighbors $n_i$ in non-ascending order.
\end{itemize}

Algorithm \ref{alg:offline} illustrates the pseudo code of pre-computing auxiliary data in detail.



\noindent\textbf{The Data Structure of Index $\mathcal{I}$}: We build a hierarchical tree index $\mathcal{I}$ over auxiliary data, which consists of leaf and non-leaf nodes.


\underline{\textit{Leaf Nodes}}: Each leaf node $\mathcal{N}$ contains multiple entries, corresponding to vertices $v_i$ and their auxiliary data $v_i.Aux$ (i.e., keyword embedding MBRs, $v_i.MBR$, and sorted lists, $v_i.list_w$, of edge weights).

\nop{

\begin{itemize}
    \item a keyword embedding MBR, $v_i.MBR$, and
    \item a sorted list of edge weights, $v_i.list_w$.
\end{itemize}

}

\underline{\textit{Non-Leaf Nodes}}: Each non-leaf node $\mathcal{N}$ contains multiple child entries $\mathcal{N}_i$, each of which is associated with the following aggregates, $\mathcal{N}_i.Aux$, over auxiliary data:
\begin{itemize}
    \item an aggregated keyword embedding MBR, $\mathcal{N}_i.MBR$, minimally bounding the MBRs, $v.MBR$, of all vertices $v$ under entry $\mathcal{N}_i$;
    \item an aggregated sorted list, $\mathcal{N}_i.list_w$, of edge weights, each element $\mathcal{N}_i.list_w [j] = \max_{\forall v \in \mathcal{N}_i} v.list_w[j]$, and;

    \item a pointer, $\mathcal{N}_i.ptr$, pointing to a child node.
\end{itemize}

$l$ is in Eq. (4) for layers, Section 5.2 for index of the sorted lists, Section 5.3, for vertex $v_l$, and Section 6.1 for subscripts of the sorted lists]



\begin{algorithm}[!t]
    \caption{\bf The Index Construction}
    \label{alg:index}
      \small
    \KwIn{
        \romannumeral1) the auxiliary data $v_i.Aux$ for each vertex $v_i\in V(G)$
        \romannumeral2) the fanout, $fanout$, of index nodes, and
        \romannumeral3) \# of iterations, $MaxIter$
    }
    \KwOut{
        a tree index, $\mathcal{I}$, over data graph $G$
    }

    initialize a tree root, $root(\mathcal{I})$, corresponding to all vertices in $V(G)$

    tree height $h=\lceil log_{fanout}(|V(G)|) \rceil$

    \For{level $x=h$ to $1$}{

        \For{each node $\mathcal{N}^{(x)}$ on $x$-th level of index $\mathcal{I}$}{


            select $fanout$ initial vertices from node $\mathcal{N}^{(x)}$ as initial centers $\{c_1, \dots, c_{fanout}\}$ for child entries $\mathcal{N}_i^{(x-1)}$ ($1 \leq i \leq fanout$) with the minimum overlapping areas of keyword embedding MBRs 
            
            
            \For{$iter = 1$ to $MaxIter$}{

                assign each vertex $v$ under $\mathcal{N}^{(x)}$ to a child entry $\mathcal{N}_i^{(x-1)}$ whose center $c_i$ results in the minimum MBR area expansion

                update centers $\{c_1, \dots, c_{fanout}\}$ based on new MBR centroids of child entries $\mathcal{N}_i^{(x-1)}$
            }


            update the aggregated auxiliary data, $\mathcal{N}_i^{(x-1)}.Aux$, in each child entry $\mathcal{N}_i^{(x-1)}$
            
        }

    }

    \Return $\mathcal{I}$
    
\end{algorithm}
\setlength{\textfloatsep}{0pt}

 
\noindent\textbf{The Index Construction}: 
We construct the index $\mathcal{I}$ in a top-down manner, as outlined in Algorithm~\ref{alg:index}. 
The root, $root(\mathcal{I})$, of index $\mathcal{I}$ is an aggregation of $V(G)$, and the height $h$ is $\lceil log_{fanout}(|V(G)|) \rceil$ (lines 1-2). 
For each non-leaf node on the $x$-th level of index $\mathcal{I}$, we partition its entries into $fanout$ child nodes on the $(x-1)$-th level based on their embedding MBRs (lines 3-9).
In particular, for each node $\mathcal{N}^{(x)}$ on the $x$-th level of index $\mathcal{I}$, we select $fanout$ initial vertices from $\mathcal{N}^{(x)}$ as initial centers $\{c_1,\cdots,c_{fanout}\}$ of $\mathcal{N}_i^{(x-1)}$ ($1 \leq i \leq fanout$) with the minimum overlapping areas of keyword embedding MBRs between each other (lines 4-5).
To optimize the partitioning, we iteratively assign each vertex in $\mathcal{N}^{(x)}$ to the child entry $\mathcal{N}_i^{(x-1)}$ that results in the least expansion of the MBR area (lines 6-7). After each assignment, we update centers based on new MBR centroids of child entries $\mathcal{N}_i^{(x-1)}$ (line 8). Once the maximum number of iterations is reached, we update each child entry $\mathcal{N}_i^{(x-1)}$, by aggregating auxiliary data, $\mathcal{N}_i^{(x-1)}.Aux$, of its vertices (line 9). 
Finally, we return the constructed tree index $\mathcal{I}$ (line 10).


\vspace{0.5ex}\noindent {\bf The Time Complexity Analysis of Offline Pre-Computations:} The time complexity of offline pre-computation mainly includes keyword embedding training, auxiliary data pre-computation, and index construction.
For the keyword embedding training, the average number of keywords per vertex is $\overline{|v_i.K|}$, building the keyword hypergraph requires $O(|V(G)|\cdot \overline{|v_i.K|})$, and the HGNN training time complexity is $O(\overline{|v_i.K|} \cdot d^{2})$.
For auxiliary data pre-computation, the computation of vertex keyword embedding MBR needs $O(|V(G)|\cdot\overline{|v_i.K|})$ and the computation of edge weight list needs $O(|V(G)|\cdot avg\_deg(G))$.
Finally, it requires $O(\lceil log_{fanout} |V(G)| \rceil \cdot fanout \cdot |V(G)|\cdot MaxIter)$ cost for the index construction.
Overall, the time complexity of offline pre-computations is given by: $O(\overline{|v_i.K|}\cdot(|V(G)|+d^{2}+avg\_deg(G))+\lceil log_{fanout} |V(G)| \rceil \cdot fanout \cdot |V(G)| \cdot MaxIter)$.


\section{Online S$^3$GND Query Processing}
\label{sec:online}

In this section, we first discuss our pruning strategies on the index node level, and then illustrate the algorithm for online S$^3$GND query answering.

\setlength{\textfloatsep}{0pt}

\subsection{Index-Level Pruning Strategies}
\noindent{\bf Index-Level Keyword Embedding MBR Pruning:}
We present \textit{index-level keyword embedding MBR pruning} to filter out node entries, $\mathcal{N}_i$, in index $\mathcal{I}$ 
whose candidate vertices do not contain all query keywords in $q_j.K$ for any query vertex $q_j \in V(q)$ (via keyword embedding MBRs).


\begin{lemma}
    \label{lemma:index_keyword_pruning}
    \textbf{(Index-Level Keyword Embedding MBR Pruning)} Given an entry $\mathcal{N}_i$ in index node $\mathcal{N}$ and a query graph $q$, entry $\mathcal{N}_i$ can be safely pruned, if it holds that: $\mathcal{N}_i.MBR \cap q_j.MBR \neq q_j.MBR$ (i.e., $q_j.MBR  \not\subseteq \mathcal{N}_i.MBR$), for all $q_j \in V(q)$.
\end{lemma}
\begin{proof}
    For all query vertices $q_j \in V(q)$, if $\mathcal{N}_i.MBR \cap q_j.MBR \neq q_j.MBR$ holds for an entry $\mathcal{N}_i$, it indicates that query keyword sets $q_j.K$ are not subsets of $v.K$ for any vertex $v \in \mathcal{N}_i$. Thus, we can infer that no vertex $v \in \mathcal{N}_i$ can match with any query vertex $q_j$. Thus, we can safely prune the entry $\mathcal{N}_i$.
\end{proof}

\noindent{\bf Index-Level ND Lower Bound Pruning:} Next, we provide the \textit{index-level ND lower bound pruning} to rule out entries, $\mathcal{N}_i$, that contain vertices with significant neighbor differences (i.e., ND scores) from query vertices $q_j\in V(q)$.

\begin{lemma}
    \label{lemma:index_ND_Lower_Bound_Pruning}
    \textbf{(Index-Level ND Lower Bound Pruning)} Given an entry $\mathcal{N}_i$ in index node $\mathcal{N}$, a query graph $q$, and a generalized neighbor difference threshold $\delta$, the entry $\mathcal{N}_i$ can be safely pruned, if it holds that $lb\_ND(q_j,\mathcal{N}_i) > \delta$ for all query vertices $q_j \in V(q)$, where the lower bound, $lb\_ND(q_j,\mathcal{N}_i)$, of the ND score for any vertex under entry $\mathcal{N}_i$ is calculated as follows:
    $$lb\_ND(q_j,\mathcal{N}_i) = \sum_{z=1}^{|N(q_j)|} \max\{q_j.list_w[z]-\mathcal{N}_i.list_w[z], 0\}.$$

\end{lemma}



\begin{proof}
    If the ND lower bound $lb\_ND(q_j,\mathcal{N}_i)$ for all $q_j \in V(q)$ is greater than threshold $\delta$, it indicates that the ND lower bound $lb\_ND(q_j, v)$ is greater than threshold $\delta$, for any query vertex $q_j \in V(q)$ and vertex $v \in \mathcal{N}_i$. Based on Definition~\ref{def:S3GND}, vertices $v$ under node $\mathcal{N}_i$ cannot match with any query vertex $q_j \in V(q)$, and thus we can safely prune entry $\mathcal{N}_i$.
\end{proof}


\nop{
{\color{blue}\noindent\textbf{Discussion on How to Obtain $lb\_ND(q_j, \mathcal{N}_i)$}: {\color{red}\bf [can you directly give it in Lemma 7?]}
}

}

\nop{\noindent{\bf Index-Level GND Lower Bound Pruning}
Similarly, we propose the \textit{index-level GND lower bound pruning} to filter out entries, $\mathcal{N}_i$, that contain vertices with significant differences from the query graph/vertex.
\begin{lemma}
    \label{lemma:index_GND_Lower_Bound_Pruning}
    \textbf{(Index-Level GND Lower Bound Pruning)} Given an entry $\mathcal{N}_i$ in index node $\mathcal{N}$, a query graph $q$, and an generalized neighbor difference threshold $\delta$, entry $\mathcal{N}_i$ can be safely pruned, if it holds that $lb\_GND(q,\mathcal{N}_i) > \delta$, where $lb\_GND(q,\mathcal{N}_i)$ is a lower bound of $GND(q,\mathcal{N}_i)$.
\end{lemma}}

\subsection{S$^3$GND Query Answering Algorithm}

Algorithm~\ref{alg:online} presents our proposed S$^3$GND query processing algorithm, which traverses the tree index $\mathcal{I}$ by applying our effective pruning strategies, and refines candidate subgraphs to retrieve actual S$^3$GND query answers.


\noindent\textbf{Initialization:} To process an S$^3$GND query, we first initialize an empty subgraph answer set, $S$  (line 1). Then, for each query vertex $q_j \in V(q)$, we use the trained keyword embedding model $\Phi$ to obtain a keyword embedding MBR $q_j.MBR$ (bounding embedding vectors of all keywords in $q_j.K$), and initialize an empty set, $q_j.V_{cand}$, of candidate vertices matching with $q_j$ (lines 2-4).




\noindent\textbf{Index Traversal:} Next, we traverse the index $\mathcal{I}$ to obtain candidate vertices in $q_j.V_{cand}$, matching with each query vertex $q_j$ (lines 5-22). Specifically, we first initialize an empty queue, $\mathcal{H}$, to facilitate the index traversal (line 5).
Moreover, we let $V(q)$ be the query vertex set, $root(\mathcal{I}).Q$, relevant to the tree root $root(\mathcal{I})$, and insert entry $root(\mathcal{I})$ into queue $\mathcal{H}$ via an enqueue operator (lines 6-7).

Each time, we pop out an entry $\mathcal{N}$ from the queue $\mathcal{H}$ via a dequeue operator (lines 8-9). 
When $\mathcal{N}$ is a leaf node, we check each vertex $v_i$ in $\mathcal{N}$. If $v_i$ cannot be pruned by Lemmas~\ref{lemma:keyword_pruning} and ~\ref{lemma:ND_lb_pruning} for some query vertex $q_j \in \mathcal{N}.Q$, then we add candidate vertex $v_i$ to the candidate vertex set $q_j.V_{cand}$ (lines 10-14).

When $\mathcal{N}$ is a non-leaf node, we check each entry $\mathcal{N}_i$ in $\mathcal{N}$, and initialize an empty set $\mathcal{N}_i.Q$, which stores query vertices that may match with some vertex in entry $\mathcal{N}_i$ (lines 15-17). Then, for each query vertex $q_j \in \mathcal{N}.Q$, if $\mathcal{N}_i$ cannot be pruned by Lemmas~\ref{lemma:index_keyword_pruning} and ~\ref{lemma:index_ND_Lower_Bound_Pruning}, we add $q_j$ to the query vertex set $\mathcal{N}_i.Q$ of entry $\mathcal{N}_i$ (lines 18-20). When $\mathcal{N}_i.Q$ is not empty, we enqueue entry $\mathcal{N}_i$ into $\mathcal{H}$ (lines 21-22). 

Finally, we assemble and refine candidate vertices in $q_j.V_{cand}$ for $q_j \in V(q)$ by invoking Algorithm~\ref{alg:refine} (applying GND lower bound pruning in Lemma~\ref{lemma:GND_Lower_Bound_Pruning}) to return final S$^3$GND query answer set $S$ (lines 23-24).

\begin{algorithm}[!t]
    \caption{\bf Online S$^3$GND Query Processing}
      \small
    \label{alg:online}
    \KwIn{
        \romannumeral1) a data graph $G$,
        \romannumeral2) a query graph $q$,
        \romannumeral3) the keyword embedding model, $\Phi$,
        \romannumeral4) index $\mathcal{I}$ over $G$,
        \romannumeral5) the aggregate function $f(\cdot)$, and 
        \romannumeral6) a generalized neighbor difference threshold $\delta$
    }
    \KwOut{
        a set, $S$, of $S^3GND$ subgraph answers similar to $q$
    }
    
    \tcp{\bf initialization}
    
    $S \leftarrow \emptyset$

    \For{each query vertex $q_j \in V(q)$}{

        obtain the keyword embedding MBR $q_j.MBR$ using model $\Phi$

        $q_j.V_{cand} \leftarrow \emptyset$
    
    }

    \tcp{\bf index traversal}

    initialize an empty queue $\mathcal{H}$ to facilitate the index traversal

    $root(\mathcal{I}).Q = V(q)$

    insert $root(\mathcal{I})$ into queue $\mathcal{H}$

    \While{$\mathcal{H}$ is not empty}{

        $\mathcal{N} = \mathcal{H}.dequeue()$

        \eIf{$\mathcal{N}$ is a leaf node}{

            \For{each vertex $v_i \in \mathcal{N}$}{

                \For{each query vertex $q_j \in \mathcal{N}.Q$}{

                    \If{$v_i$ cannot be pruned by Lemmas~\ref{lemma:keyword_pruning} and ~\ref{lemma:ND_lb_pruning}}{

                        $q_j.V_{cand} \leftarrow q_j.V_{cand} \cup \{v_i\}$
                    
                    }
                    
                }
            
            }
        }{

        \tcp{$\mathcal{N}$ is a non-leaf node}

            \For{each entry $\mathcal{N}_i \in \mathcal{N}$}{

                $\mathcal{N}_i.Q \leftarrow \emptyset$

                \For{each query vertex $q_j \in \mathcal{N}.Q$}{

                    \If{$\mathcal{N}_i$ cannot be pruned by Lemmas~\ref{lemma:index_keyword_pruning} and ~\ref{lemma:index_ND_Lower_Bound_Pruning}}{

                        $\mathcal{N}_i.Q \leftarrow \mathcal{N}_i.Q \cup \{q_j\}$
                    
                    }
                }

                \If{$\mathcal{N}_i.Q$ is not empty}{


                    $\mathcal{H}.enqueue(\mathcal{N}_i)$
                }
            
            }
        }

    }

    \tcp{\bf refinement}
    invoke Algorithm~\ref{alg:refine} to assemble and refine the candidate vertices in $q_j.V_{cand}$ for $q_j \in V(q)$ with Lemma~\ref{lemma:GND_Lower_Bound_Pruning} to obtain similar subgraphs $g$ form $S$
    
    \Return $S$
\end{algorithm}

\begin{algorithm}[t]
    \caption{\bf Refinement}
    \label{alg:refine}
    \small
    \KwIn{
        \romannumeral1) a data graph $G$,
        \romannumeral2) a query graph $q$,
        \romannumeral3) a candidate vertex list of the query vertices, $q_j.V_{cand}$ $(1 \leq j \leq |V(q)|)$, and
        \romannumeral5) a generalized neighbor difference threshold $\delta$
        }
    \KwOut{
        a set, $S$, of $S^3GND$ subgraph answers similar to $q$
    }

    $S \leftarrow \emptyset, stack \leftarrow \emptyset$

    $q_{start} \leftarrow \arg\min_{j}(|q_j.V_{cand}|)$

    \For{each candidate vertex $v_i \in q_{start}.V_{cand}$}{

        $initial\_mapping \leftarrow \{q_{start}:v_i\}$

        $initial\_used \leftarrow \{v_i\}$

        $initial\_remaining \leftarrow V(q)-\{q_{start}\}$

        $stack.push(initial\_mapping, initial\_used, initial\_remaining)$
    }

    \While{$stack$ is not empty}{
        
        $(curr\_{map}, curr\_used, curr\_rem) \leftarrow stack.pop()$

        \If{$|curr\_rem|=0$}{

            \If{$curr\_{map}$ cannot pruned by Lemma~\ref{lemma:GND_Lower_Bound_Pruning}}{
            
                \If{$GND(q, curr\_{map})\leq \delta$}{
    
                    $S.append(curr\_{map})$
                }
            }
            
        }

        \Else
        {
    
            $q_{next} \leftarrow \max\_neighbors(curr\_rem)$

            $new\_rem \leftarrow curr\_rem- \{q_{next}\}$

            \For{each $v_i \in q_{next}.V_{cand}$}{

                \If{$v_i \notin curr\_used$}{


                    $new\_map \leftarrow curr\_map \cup \{q_{next}:v_i\}$ 

                    $new\_used\leftarrow curr\_used \cup \{v_i\}$

                    $stack.push(new\_map, new\_used, new\_rem)$
                }
                
            }

        }

    }
    
    \Return $S$
\end{algorithm}

\noindent\textbf{Refinement:} Algorithm~\ref{alg:refine} details our refinement process, which assembles candidate vertices into subgraphs and refines them.
Specifically, we first initialize an empty S$^3$GND query answer set $S$ and an empty stack, $stack$, and then choose a query vertex $q_{start}$ with the smallest set, $q_{start}.V_{cand}$, of candidate vertices, to start the search (lines 1-2). 
For each candidate vertex $v_i$ in the candidate vertex set $q_{start}.V_{cand}$, we initialize an entry consisting of three variables, that is, ($initial\_mapping$, $initial\_used$, $initial\_remaining$). 
Here, $initial\_mapping$ stores the matching relationship between the query vertex and the candidate vertex and initialized to $\{q_{start}:v_i\}$, $initial\_used$ records the usage status of the candidate vertex and is initialized to $\{v_i\}$, and $initial\_remaining$ contains query vertices that have not yet been checked and is initialized to $V(q)-\{q_{start}\}$ (lines 3-6). Then, we push this entry into the stack, $stack$, for the graph traversal (line 7).

Each time, we pop out an entry $(curr\_map, curr\_used, curr\_rem)$ from stack $stack$ (lines 8-9). 
If the number of currently unmatched query vertices in $curr\_rem$ is 0 (i.e., all query vertices have been matched), we apply Lemma~\ref{lemma:GND_Lower_Bound_Pruning} to prune the candidate subgraph with vertices in  $curr\_{map}$ (line 10). If the matching in $curr\_{map}$ cannot be pruned, we compute the actual GND score, $GND(q, curr\_{map})$, and add a candidate subgraph with vertices in $curr\_{map}$ to the candidate answer set $S$ (lines 11-13).


If $curr\_rem$ is not zero (i.e., not all query vertices have been matched), we select a query vertex from $curr\_rem$ with the largest adjacency weight among the currently matched query vertices in $curr\_map$ as the next query vertex $q_{next}$ to ensure the vertex connectivity of $curr\_map$ (i.e., candidate subgraph).  We then remove $q_{next}$ from $curr\_{rem}$, resulting in a new remaining set $new\_rem$ (lines 14-16). 

\nop{
we select a neighbor vertex with the 
{\color{red}highest edge weight \bf [what is this weight? give rationale; make it clear]} from $curr\_{rem}$ to serve as the next query vertex $q_{next}$.
}

Next, we evaluate each vertex $v_i$ in $q_{next}.V_{cand}$ (line 17). If $v_i$ is not already included in $curr\_used$, we create a new mapping set $new\_map$, which includes $curr\_map$ and one new mapping $\{q_{next}:v_i\}$ (lines 18-19). We also obtain a new used set $new\_used$ by unioning $curr\_used$ with $v_i$ (line 20). We then push the new entry $(new\_map, new\_used, new\_rem)$ into the stack, $stack$, for subsequent traversal (line 21). Finally, we return all S$^3$GND query results in $S$, after the stack, $stack$, becomes empty (line 22).

\nop{
we check whether or not $curr\_map$ can be pruned by {\bf Lemma~\ref{lemma:GND_Lower_Bound_Pruning} {\color{red}[why Lemma 3?]}}. If the answer is no, 
{\color{orange} {\bf [not very readable here! please rewrite this part]} 
then we select a neighbor vertex with the highest weight that has matched the 1-hop of currently matched query vertices as the next query vertex $next$ and remove the $next$ in $curr\_rem$ as a new remaining set $new\_rem$ of $next$ (lines 13-16). Next, we check each vertex $v_i$ in $next.V_{cand}$. If $v_i$ is not in $curr\_used$, we create a new mapping set $new\_map$ identical to $curr\_map$ and add $\{next:v_i\}$ to it. We add $v_i$ to the used set and pass it to a new vertex used set $new\_used$ (lines 17-21).} 

{\color{red}\bf [for refinement, i changed the if statement to if-else statement in lines 10-13. 

Also I changed from one single if to two if statements in lines 10-11. please make sure this is correct.

please check your Python source code to see if the code is correct.]}

}

\vspace{0.5ex}\noindent {\bf The Time Complexity Analysis of Online S$^3$GND Query Algorithm:} The time complexity of the initialization is $O(\sum_{j=1}^{|V(q)|} |q_j.K|)$ in Algorithm~\ref{alg:online}. Let $PP_i$ denote the pruning power (i.e., the percentage of node entries that can be pruned) on the $i$-th level of the tree index $\mathcal{I}$, where $1 \leq i \leq \lceil log_{fanout} |V(G)| \rceil$.
Then, for the index traversal process, the number of visited nodes is $\sum_{i=1}^{\lceil log_{fanout} |V(G)| \rceil} fanout \cdot (1-PP_i)$.
The update cost of each candidate vertex set and the computation of the ND lower bound is $O(1)$.
For the refinement, the time complexity of initializing the stack is $O(min(\vee_{q_j \in V(q)}|q_j.V_{cand}|))$. Let $PP_j$ denote the GND lower bound pruning power on the $j$-th query candidate vertex set. The refinement needs $O(\prod_{j=1}^{|V(q)|} |q_j.V_{cand}|\cdot (1-PP_j))$.
Therefore, the overall time complexity of online S$^3$GND query processing (i.e.,  Algorithm~\ref{alg:online}) is given by $O\big(\sum_{j=1}^{|V(q)|} |q_j.K| + \\ \sum_{i=1}^{\lceil log_{fanout} |V(G)| \rceil}  fanout \cdot (1-PP_i) +\prod_{j=1}^{|V(q)|} |q_j.V_{cand}|\cdot (1-PP_j)\big)$.

\section{Experimental Evaluation}
\label{sec:experiment}

\subsection{Experimental Settings}
\textbf{Real/Synthetic Graph Data sets}: We use both real and synthetic graph data sets to evaluate our S$^3$GND query algorithm.


\underline{\it Real-World Graph Data Sets}: We use five real-world graph data sets widely used in the literature~\cite{liang2025towards, zitnik2017predicting, zenggraphsaint, yang2023pane, pan2018adversarially, zhou2018prre}, including citation networks (e.g., $Cora$ and $PubMed$~\cite{yang2023pane}), city networks (e.g., $shanghai$~\cite{liang2025towards}), and social networks (e.g., $Wiki$ and $TWeibo$~\cite{yang2023pane}). Statistics of these real graphs are summarized in Table~\ref{tab:datasets}.

\underline{\it Synthetic Graph Data Sets}: In accordance with prior works~\cite{pavan2013counting, wen2025s3and, ye2024efficient, lu2017finding}, we generate synthetic small-world graph data sets via the Newman-Watts-Strogatz model~\cite{watts1998collective}, using NetworkX~\cite{hagberg2007exploring} and igraph~\cite{csardi2006igraph} libraries. With the generated graph structures, we obtain three synthetic graphs: \textit{Syn-Uni}, \textit{Syn-Gau}, and \textit{Syn-Zipf}, with different keyword distributions \textit{Uniform}, \textit{Gaussian}, and \textit{Zipf}, respectively, over the keyword domain $\Sigma$.


\begin{table}[t]\vspace{-2ex}
\begin{center}
\caption{\small Statistics of Real-World Graph Data Sets.}
\vspace{-2ex}
\label{tab:datasets}
\scriptsize
\begin{tabular}{|c||c||c||c||c||c|}
\hline
\textbf{Name} & \textbf{Abbr.} & $\bm{|V(G)|}$ & $\bm{|E(G)|}$ & $\bm{|\sum|}$ & \textbf{Avg. Deg.}\\
\hline\hline
    $Cora$ & CR & 2,708 & 5,429 & 1,433 & 4.01\\\hline
    $Wiki$ & WK & 2,405 & 17,981 & 4,973 & 14.95\\\hline
    $PubMed$ & PM & 19,717 & 44,338 & 500 & 4.50\\\hline
    $Shanghai$ & SH & 183,917 & 524,184 & 37 & 5.41\\\hline
    $TWeibo$ & TW & 2,320,895 & 9,840,066 & 1,658 & 8.48\\\hline
\end{tabular}
\end{center}\vspace{-2ex}
\end{table}

\begin{table}[t]\vspace{-1ex}
\begin{center}
\caption{\small Parameter Settings.}
\label{tab:parameters}
\vspace{-2ex}
\scriptsize
\resizebox{\linewidth}{!}{
\begin{tabular}{|p{5cm}||l|}
\hline
\textbf{Parameters}&\textbf{Values} \\
\hline\hline
    the threshold, $\delta_{MAX}$ ($=$ $\delta$), of MAX neighbor difference & {\bf1}, 2, 3, 4 \\\hline
    the threshold, $\delta_{SUM}$ ($=$ $\delta$), of SUM neighbor difference & 2, {\bf 3}, 4, 5 \\\hline
    the size, $|v_i.W|$, of keywords per vertex & 1, 2, {\bf 3}, 4, 5 \\\hline
    the size, $|\Sigma|$, of the keyword domain & 10, 20, {\bf50}, 80 \\\hline
    the size, $|V(q)|$, of query graph $q$ & 3, {\bf 5}, 8, 10 \\\hline
    the size, $|V(G)|$, of data graph $G$ & 10K, {\bf 50K}, 250K, 1M, 10M, 30M \\\hline    
\end{tabular}
}
\end{center}
\end{table}

\noindent{\bf Competitors}: To the best of our knowledge, no previous work provides an exact solution to our S$^3$GND problem (i.e., subgraph similarity search over a weighted data graph under GND semantics, considering vertex keyword sets and edge weights). Alternatively, we choose S$^3$AND~\cite{wen2025s3and} and Bloom Filter (BF)~\cite{xie2021hash} algorithms as comparative baselines. Specifically, S$^3$AND~\cite{wen2025s3and} cannot handle GND semantics in graphs with edge weights, so we compare our S$^3$GND approach with S$^3$AND over a weighted graph without using the pruning w.r.t. edge existences (i.e., AND scores~\cite{wen2025s3and}).
For BF, we use bloom filters to check the existence of query keyword sets and obtain candidate vertices. For a fair comparison, we set the size of a bloom filter the same as that of the keyword embedding MBR (i.e., 2$d$) in our S$^3$GND approach.






\noindent{\bf Query Graph}: We extract the connected subgraphs from the data graph $G$, by using random walks~\cite{xia2019random, wen2025s3and}. Then, we also sample keywords in the vertices of these subgraphs with a sampling rate of $0.9$, and sample edges inside subgraphs with an edge sampling rate of $0.7$, resulting in connected query graphs $q$.


\noindent{\bf Measures:} To highlight the effectiveness and efficiency of our proposed S$^3$GND approach, we evaluate the \textit{pruning power} and \textit{wall clock time}, which are defined as the percentage of candidate vertices that can be filtered out and the CPU time of S$^3$GND query processing, respectively. For each experiment, we report the average of 50 S$^3$GND queries.


\noindent{\bf Parameter Settings:} Table~\ref{tab:parameters} depicts the parameter settings of graph data sets and our S$^3$GND query predicates, where default values are in bold. Each time we vary one parameter, and set other parameters to their default values. By default, we set the embedding dimension, $d$, to 64 and the node fanout, $fanout$, of index $\mathcal{I}$ to 16. All experiments were conducted on a PC equipped with an NVIDIA RTX PRO 6000 featuring 96GB of video memory and an Intel Core i9-14900K with 192GB of RAM. All algorithms were implemented using Python 3.11. The source code and data sets are available at URL: {\it\url{https://github.com/Luminous-wq/S3GND}}.


\subsection{S$^3$GND Effectiveness Evaluation}

\begin{figure}[t]
    \centering\vspace{-2ex}
    \subfigure[real-world graphs]{
            \includegraphics[width=0.48\linewidth]{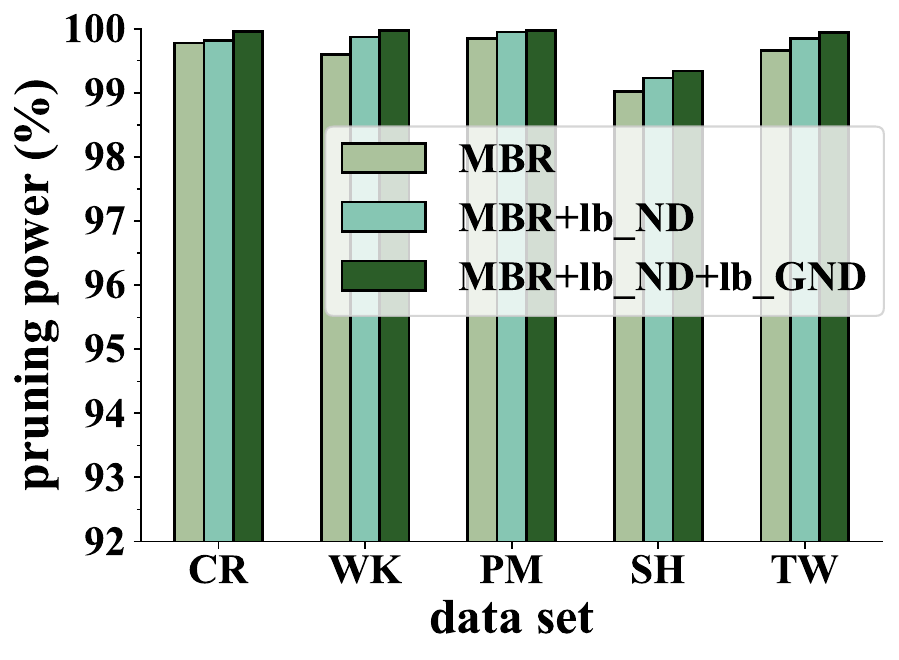}
            \label{subfig:S3GND_real_pruning}
    } \hspace{-0.2cm}
    \subfigure[synthetic graphs]{
        \includegraphics[width=0.48\linewidth]{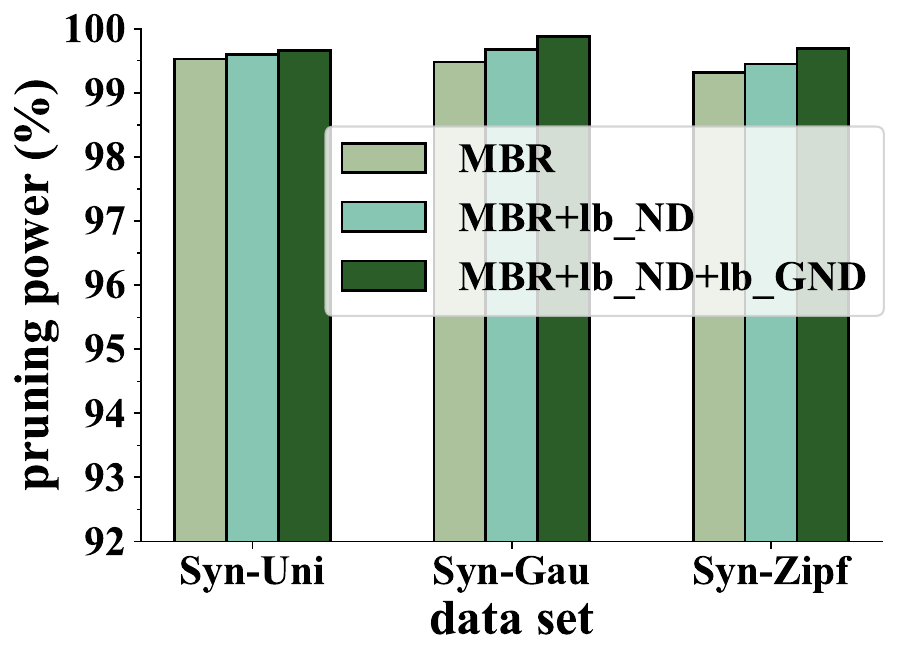}
        \label{subfig:S3GND_synthetic_pruning}
    }\vspace{-2ex}
    \caption{Evaluation of the S$^3$GND effectiveness.}
    \label{fig:pruning}
    \vspace{-1ex}
\end{figure}


Figure~\ref{fig:pruning} evaluates the effectiveness of our S$^3$GND approach via an ablation study, in terms of the \textit{pruning power}. Specifically, we test three combinations of pruning strategies: (1) \textit{keyword embedding MBR pruning} (denoted as MBR); (2) \textit{keyword embedding MBR pruning + ND lower bound pruning} (denoted as MBR+$lb\_ND$), and; (3) \textit{keyword embedding MBR pruning + ND lower bound pruning + GND lower bound pruning} (denoted as MBR+$lb\_ND$+$lb\_GND$), where all parameters are set to default values.
From experimental results, different combinations of pruning methods can achieve high pruning power (over 99\%) on both real and synthetic graph data sets. Moreover, applying additional pruning methods increases the pruning power. Our S$^3$GND approach (with all three pruning strategies) can reach $99.02\% \sim 99.99\%$ pruning power for real graphs, and $99.32\% \sim 99.88\%$ for synthetic graphs, which confirms the effectiveness of our pruning strategies.

\begin{figure}[!t]
    \centering\vspace{-2ex}
    \subfigure[real-world graphs]{
        \includegraphics[width=0.57\linewidth]{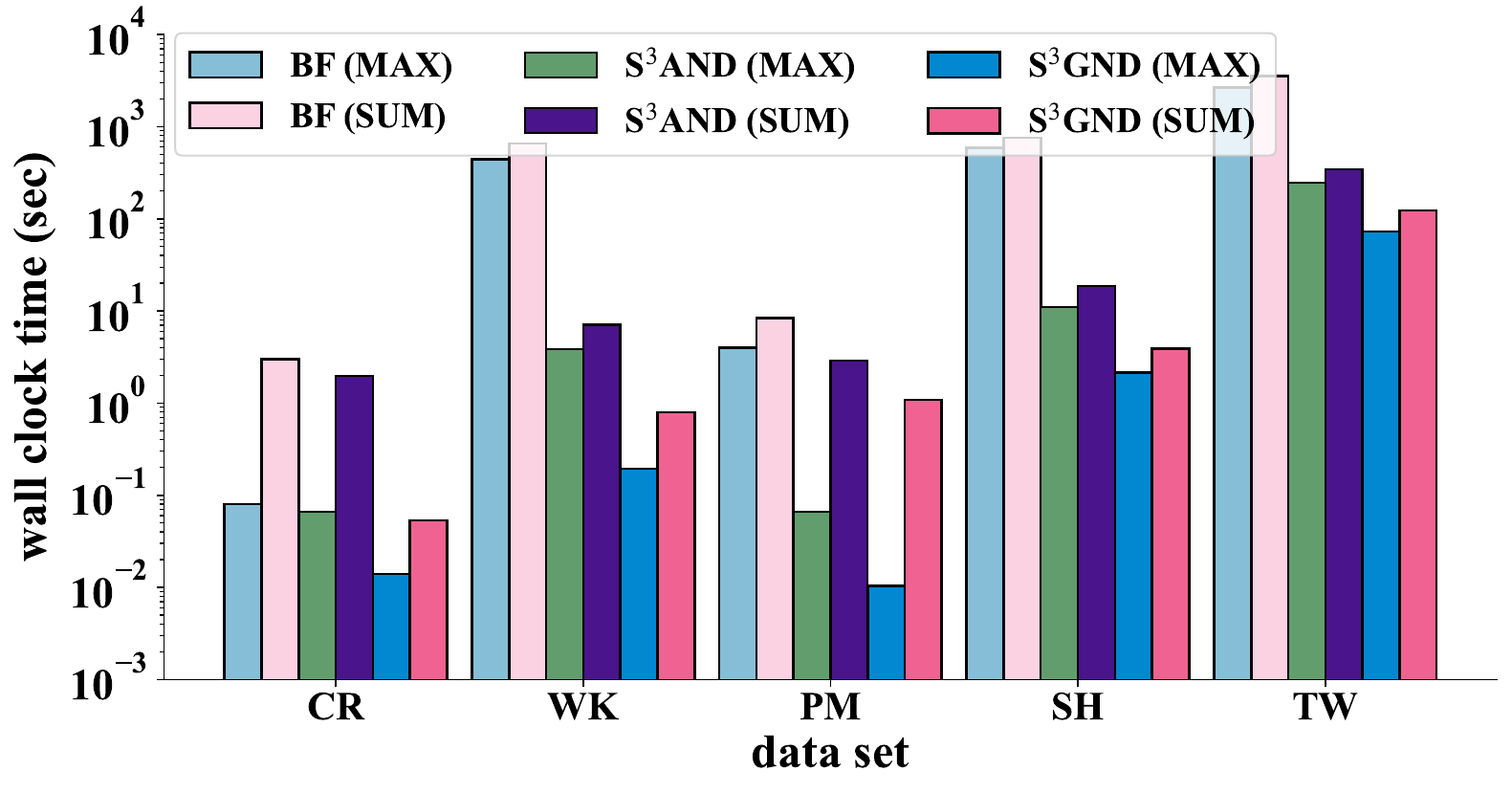}
            \label{subfig:bases_online_real}
    } \hspace{-0.3cm}
    \subfigure[synthetic graphs]{
        \includegraphics[width=0.41\linewidth]{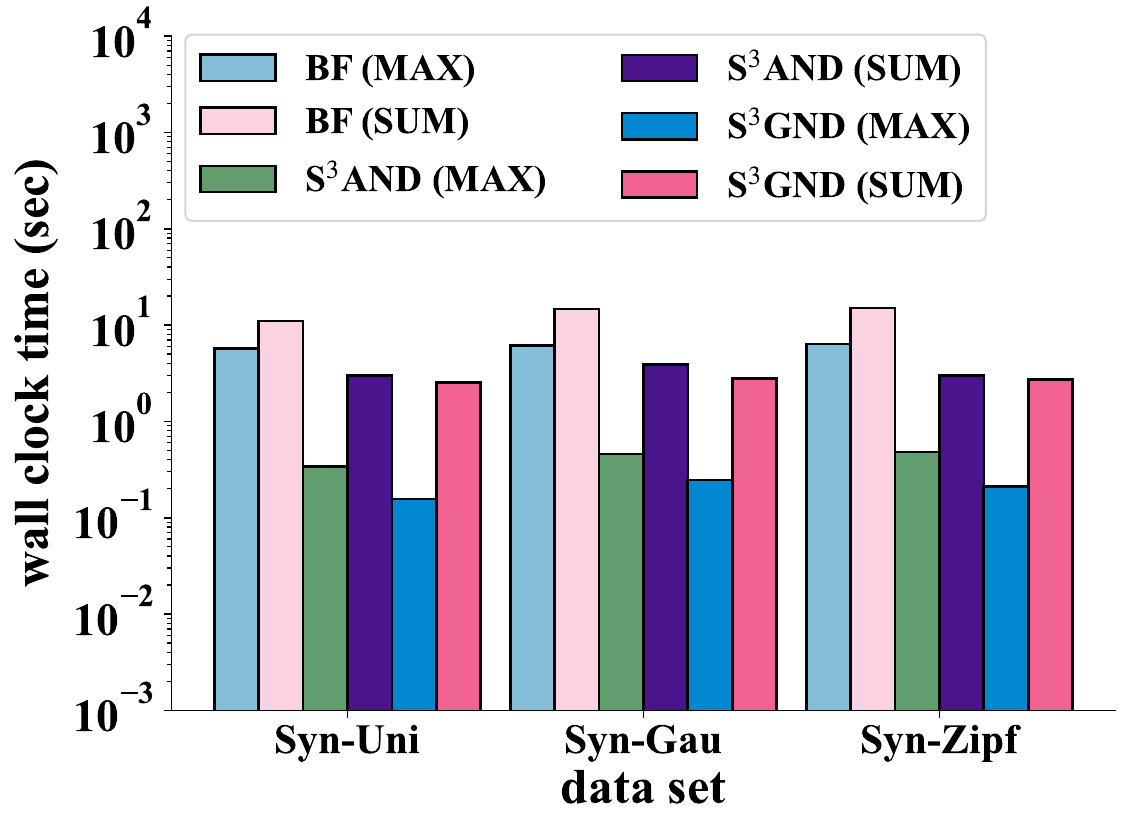}
        \label{subfig:bases_online_syn}
    }
    \vspace{-2ex}
    \caption{Comparison of our S$^3$GND approach with $BF$ and S$^3$AND baseline methods over real and synthetic graphs.}
    \label{fig:baseline}
    \vspace{0.5ex}
\end{figure}


\begin{figure*}[t]
    \centering
    \subfigure[{\scriptsize time vs. $\delta_{MAX}$ ($Syn\text{-}Uni$)}]{
            \includegraphics[width=0.15\linewidth]{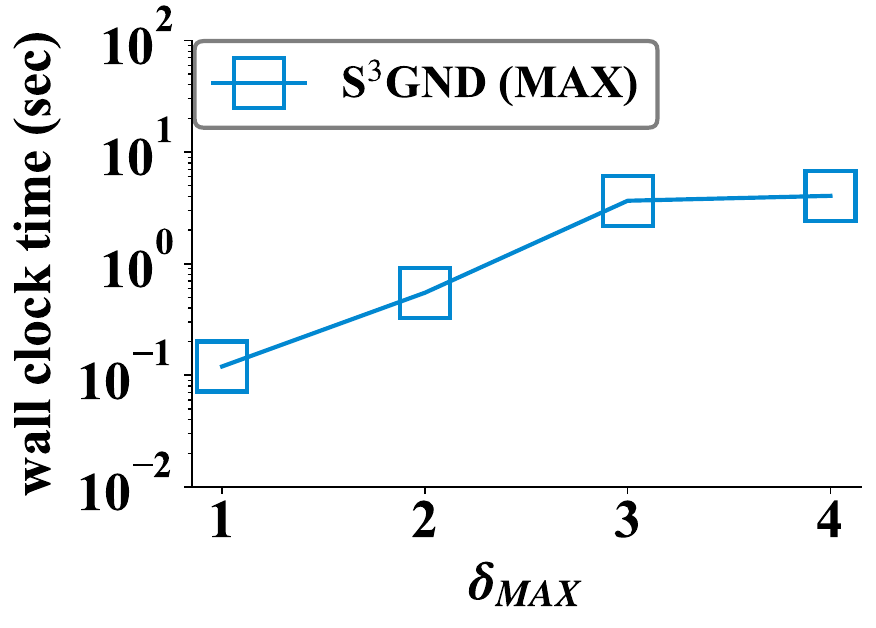}
            \label{subfig:uni-max}
    }
    \subfigure[{\scriptsize time vs. $\delta_{SUM}$ ($Syn\text{-}Uni$)}]{
            \includegraphics[width=0.15\linewidth]{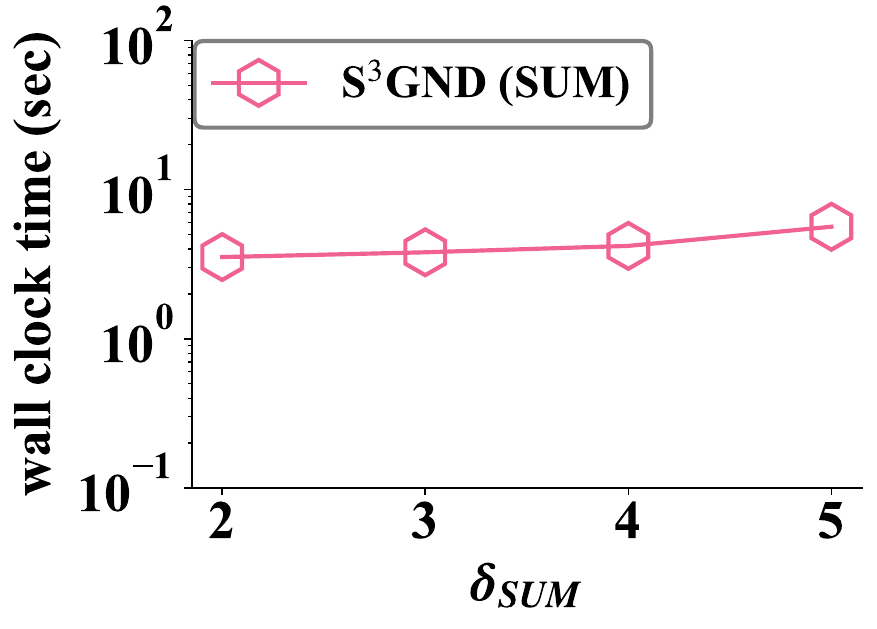}
            \label{subfig:uni-sum}
    }
    \subfigure[{\scriptsize time vs. $|v_i.W|$ ($Syn\text{-}Uni$)}]{
            \includegraphics[width=0.15\linewidth]{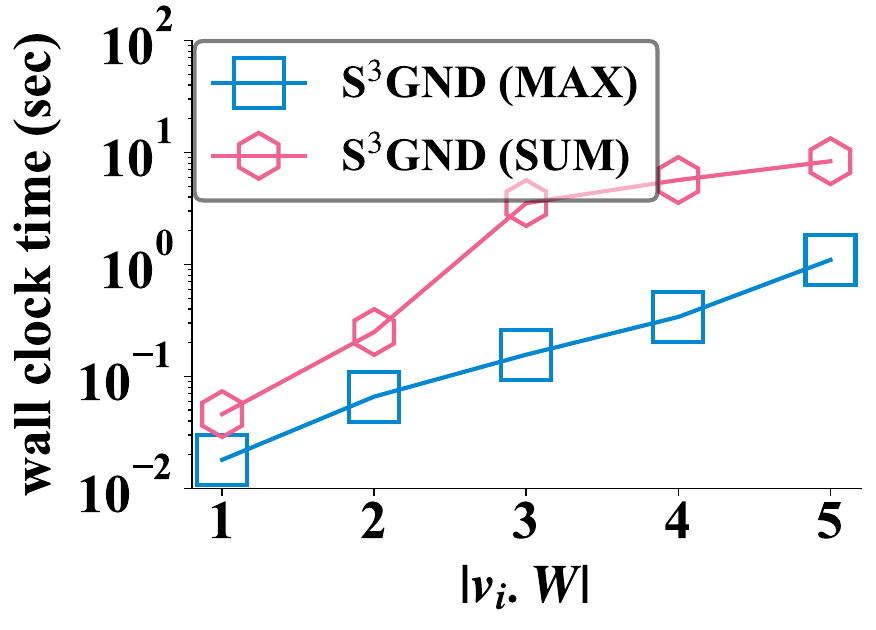}
            \label{subfig:uni-w}
    }
    \subfigure[{\scriptsize time vs. $|\Sigma|$ ($Syn\text{-}Uni$)}]{
            \includegraphics[width=0.15\linewidth]{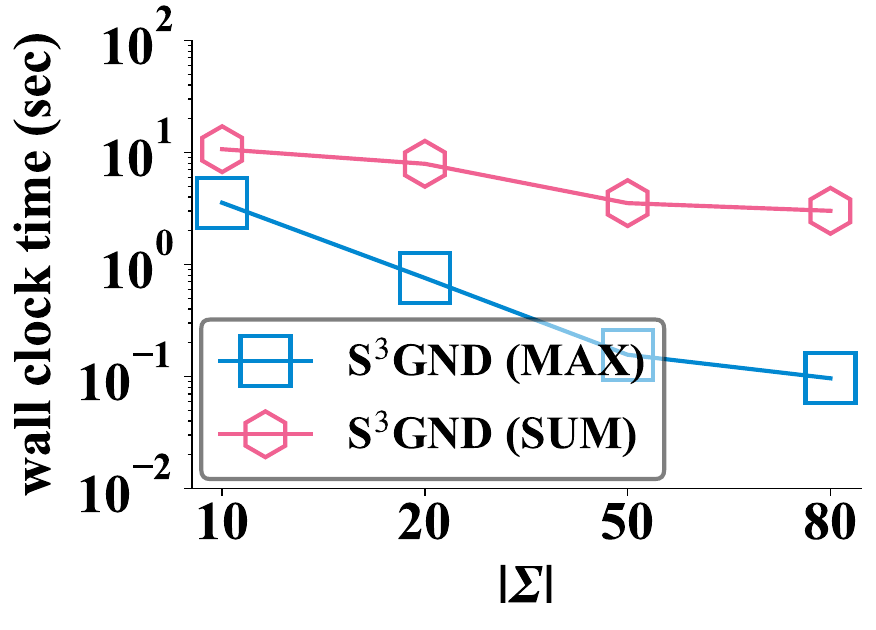}
            \label{subfig:uni-domain}
    }
    \subfigure[{\scriptsize time vs. $|V(q)|$ ($Syn\text{-}Uni$)}]{
            \includegraphics[width=0.15\linewidth]{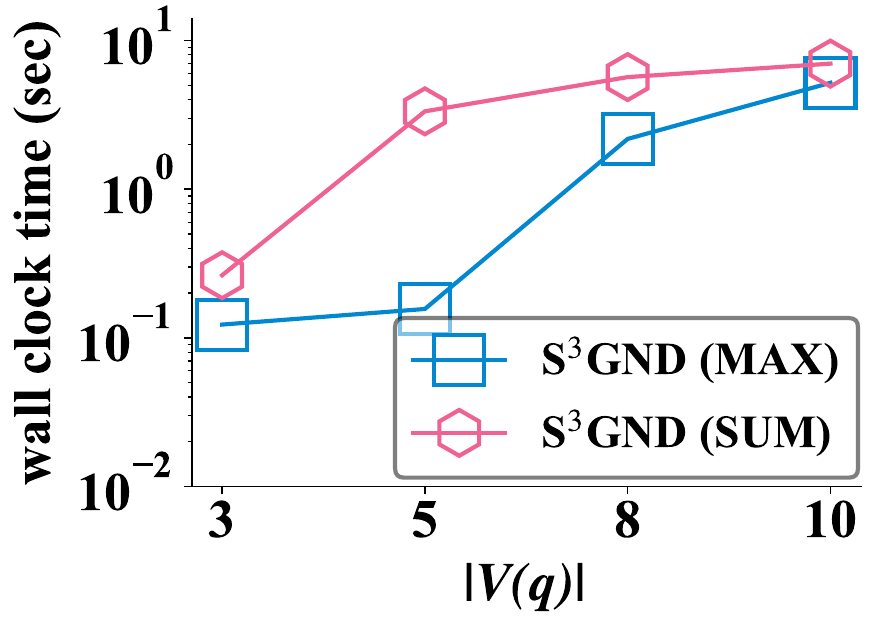}
            \label{subfig:uni-vq}
    }
    \subfigure[{\scriptsize time vs. $|V(G)|$ ($Syn\text{-}Uni$)}]{
            \includegraphics[width=0.15\linewidth]{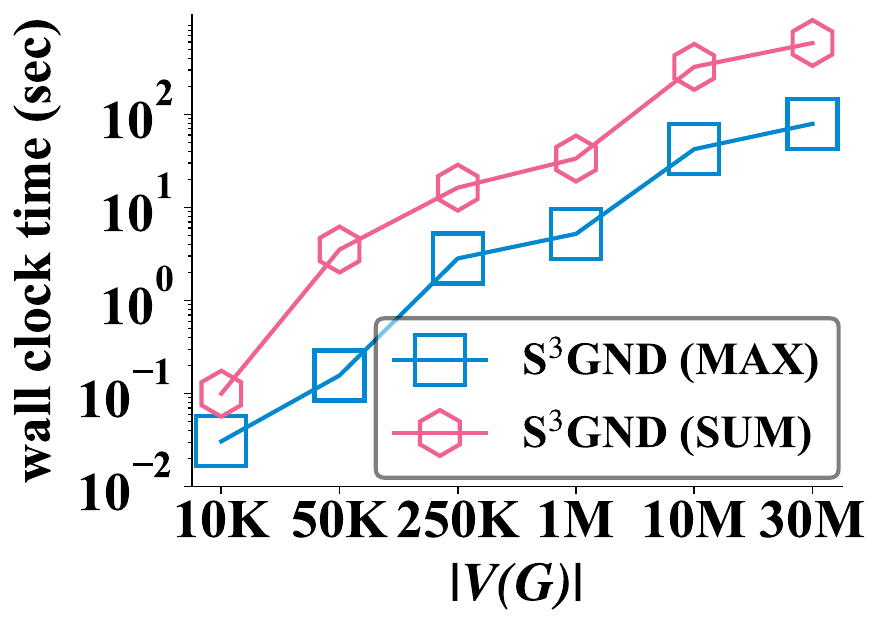}
            \label{subfig:uni-vG}
    }
    
    \subfigure[{\scriptsize time vs. $\delta_{MAX}$ ($Syn\text{-}Gau$)}]{
            \includegraphics[width=0.15\linewidth]{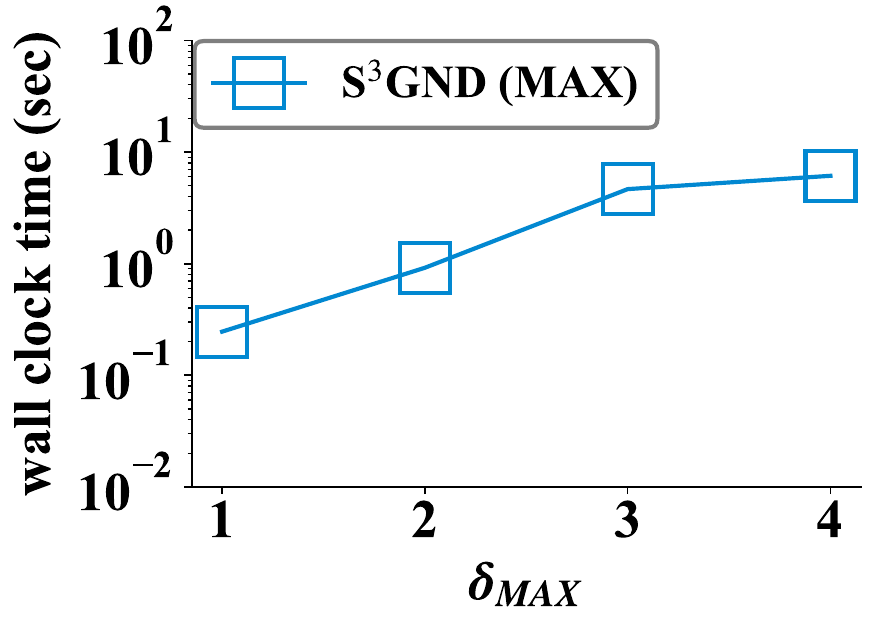}
            \label{subfig:gau-max}
    }
    \subfigure[{\scriptsize time vs. $\delta_{SUM}$ ($Syn\text{-}Gau$)}]{
            \includegraphics[width=0.15\linewidth]{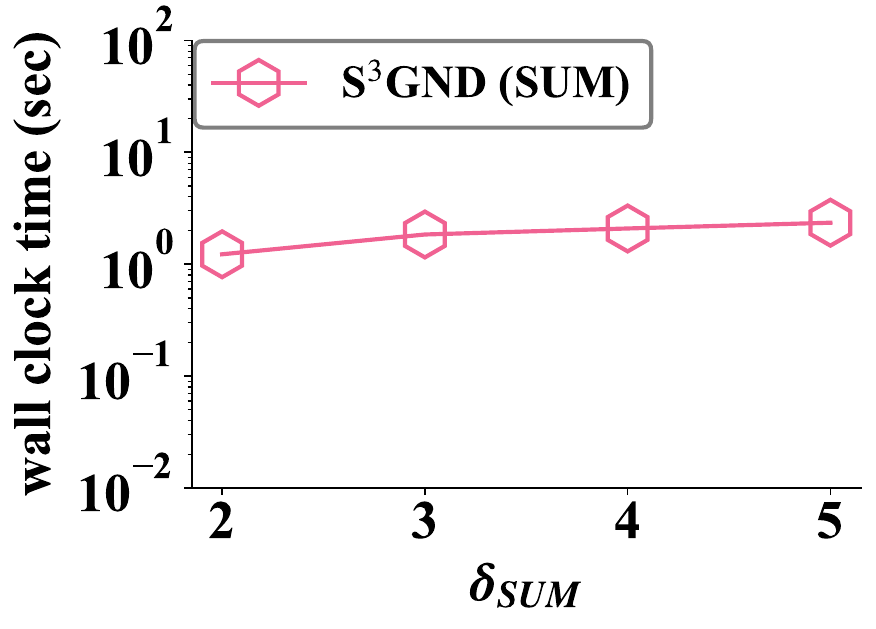}
            \label{subfig:gau-sum}
    }
    \subfigure[{\scriptsize time vs. $|v_i.W|$ ($Syn\text{-}Gau$)}]{
            \includegraphics[width=0.15\linewidth]{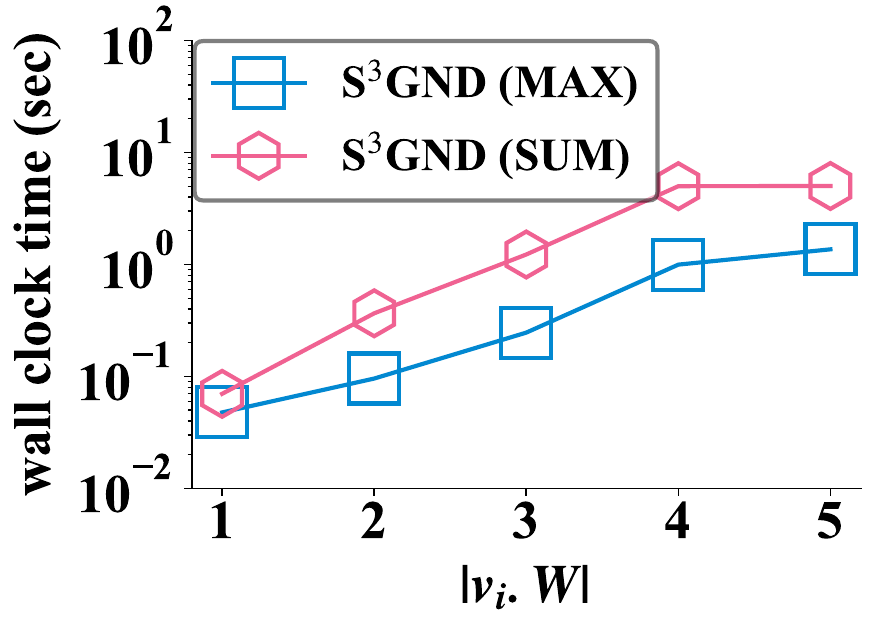}
            \label{subfig:gau-w}
    }
    \subfigure[{\scriptsize time vs. $|\Sigma|$ ($Syn\text{-}Gau$)}]{
            \includegraphics[width=0.15\linewidth]{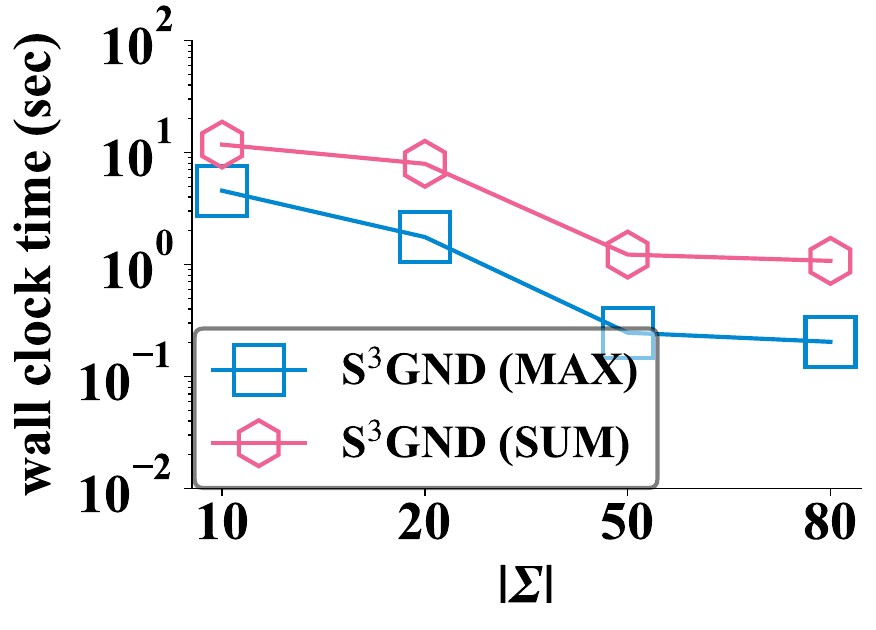}
            \label{subfig:gau-domain}
    }
    \subfigure[{\scriptsize time vs. $|V(q)|$ ($Syn\text{-}Gau$)}]{
            \includegraphics[width=0.15\linewidth]{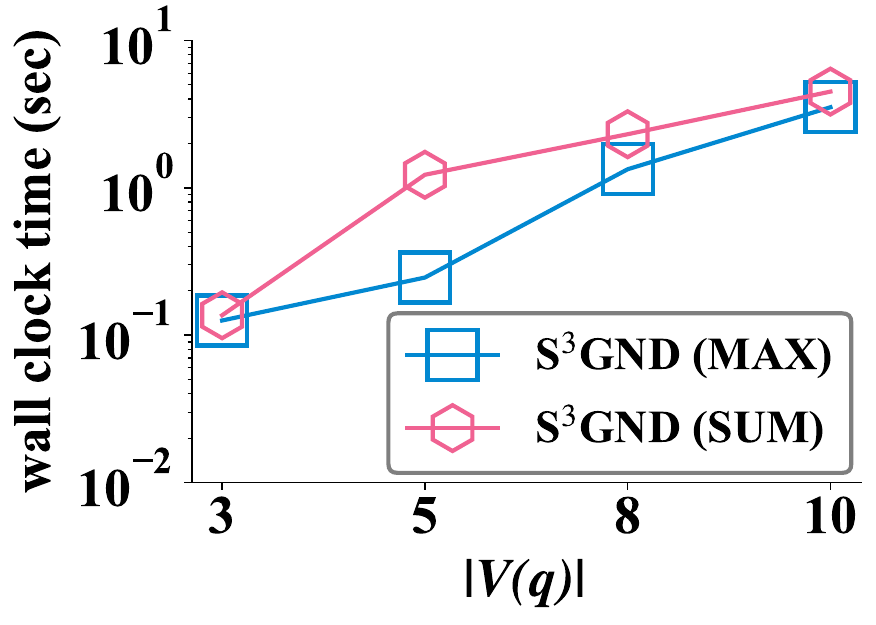}
            \label{subfig:gau-vq}
    }
    \subfigure[{\scriptsize time vs. $|V(G)|$ ($Syn\text{-}Gau$)}]{
            \includegraphics[width=0.15\linewidth]{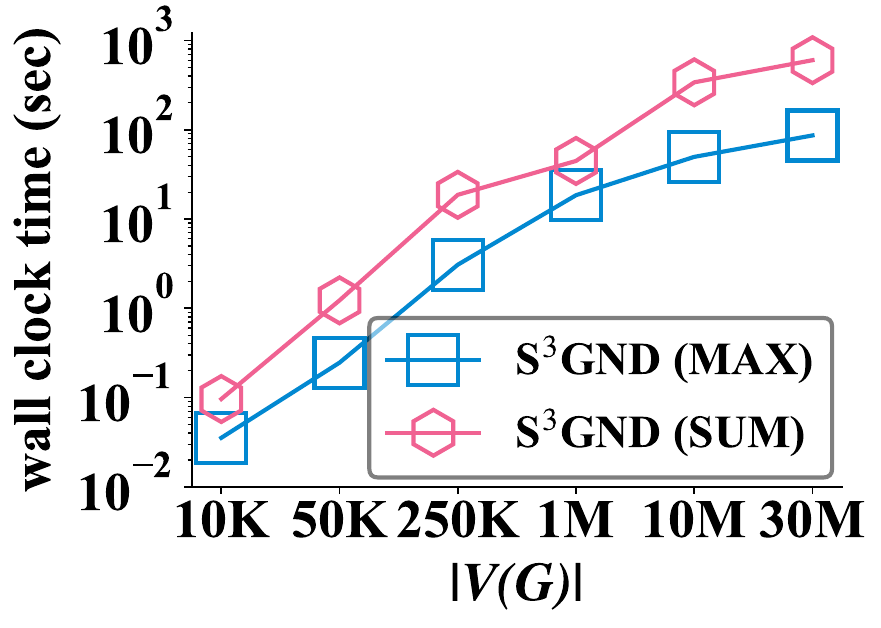}
            \label{subfig:gau-vG}
    }

    \subfigure[{\scriptsize time vs.$\delta_{MAX}$($Syn\text{-}Zipf$)}]{
            \includegraphics[width=0.15\linewidth]{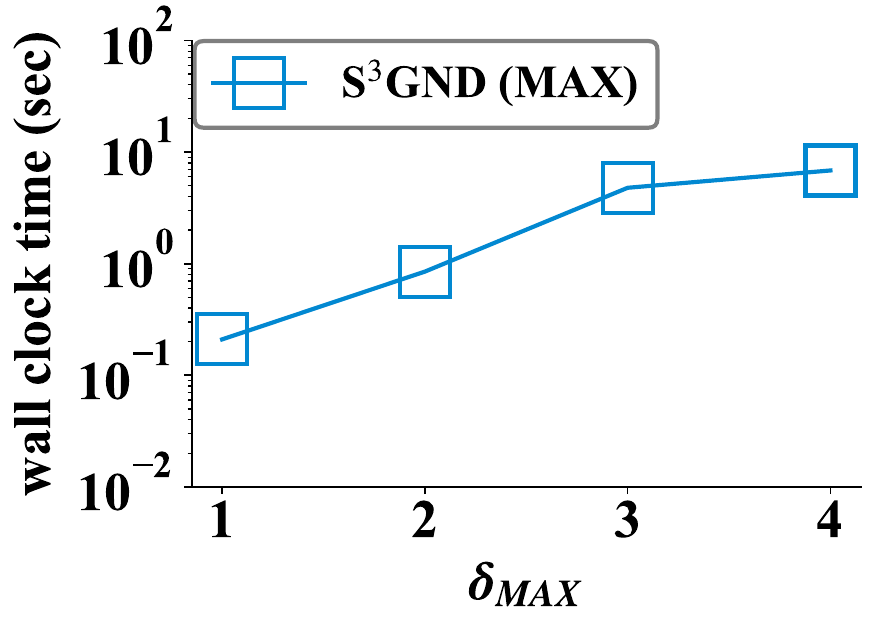}
            \label{subfig:zipf-max}
    }
    \subfigure[{\scriptsize time vs. $\delta_{SUM}$ ($Syn\text{-}Zipf$)}]{
            \includegraphics[width=0.15\linewidth]{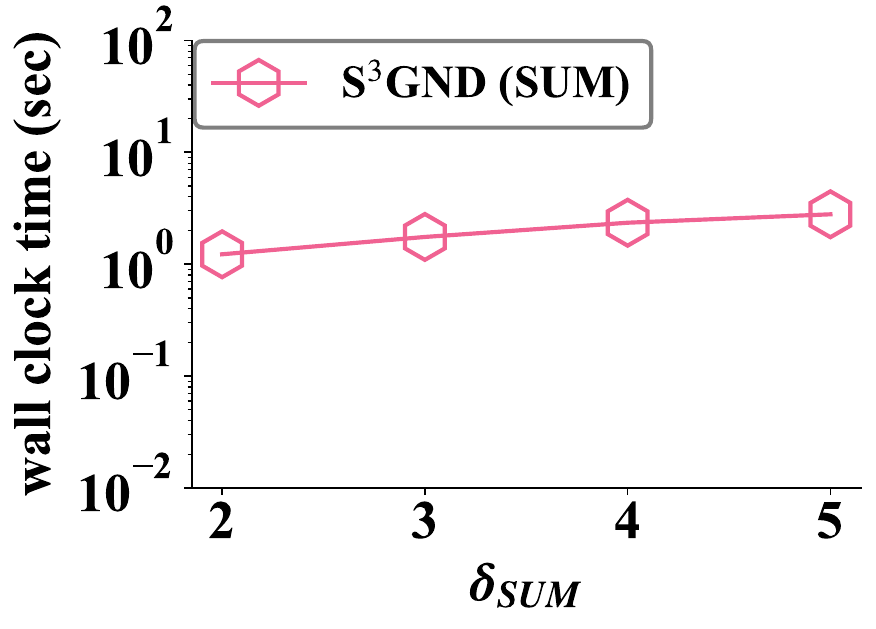}
            \label{subfig:zipf-sum}
    }
    \subfigure[{\scriptsize time vs. $|v_i.W|$ ($Syn\text{-}Zipf$)}]{
            \includegraphics[width=0.15\linewidth]{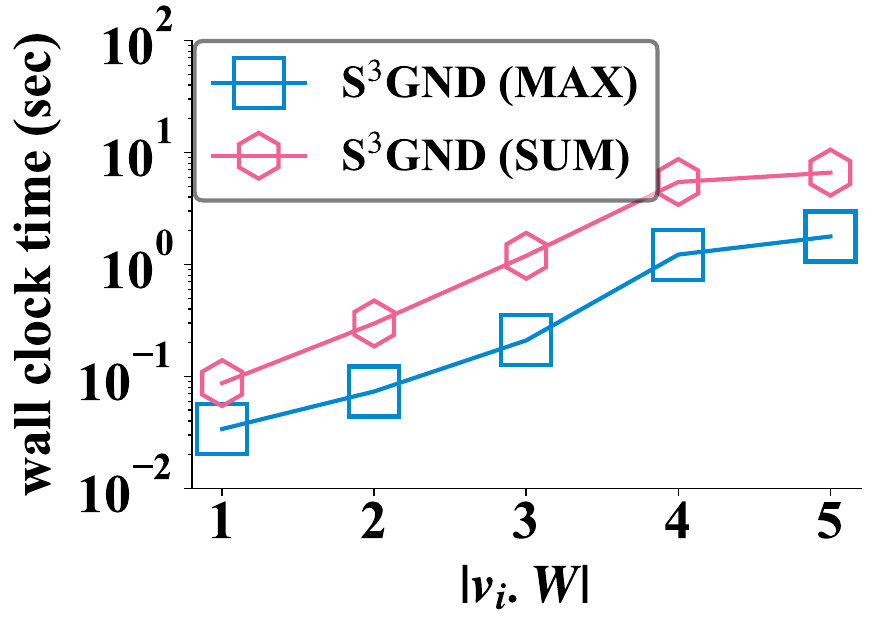}
            \label{subfig:zipf-w}
    }
    \subfigure[{\scriptsize time vs. $|\Sigma|$ ($Syn\text{-}Zipf$)}]{
            \includegraphics[width=0.15\linewidth]{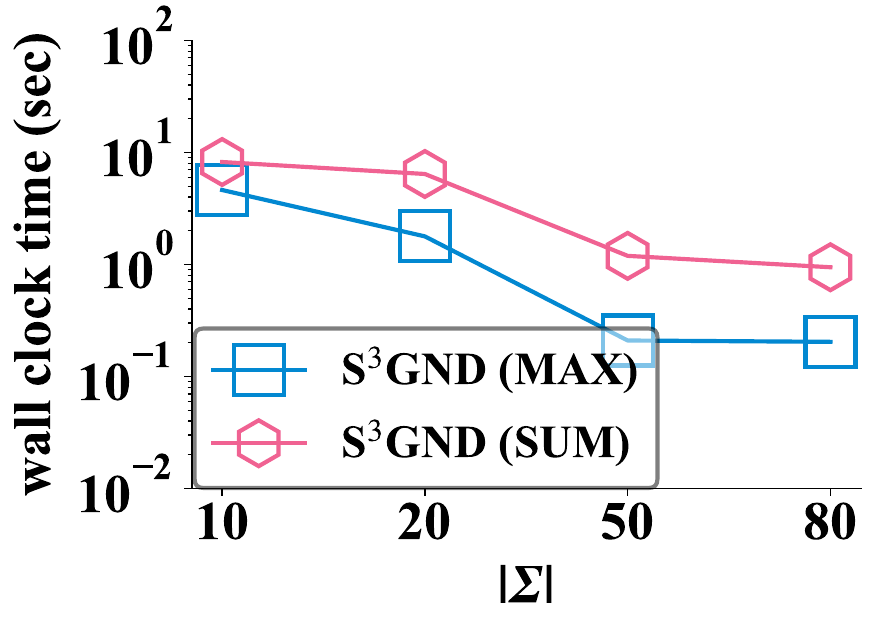}
            \label{subfig:zipf-domain}
    }
    \subfigure[{\scriptsize time vs. $|V(q)|$ ($Syn\text{-}Zipf$)}]{
            \includegraphics[width=0.15\linewidth]{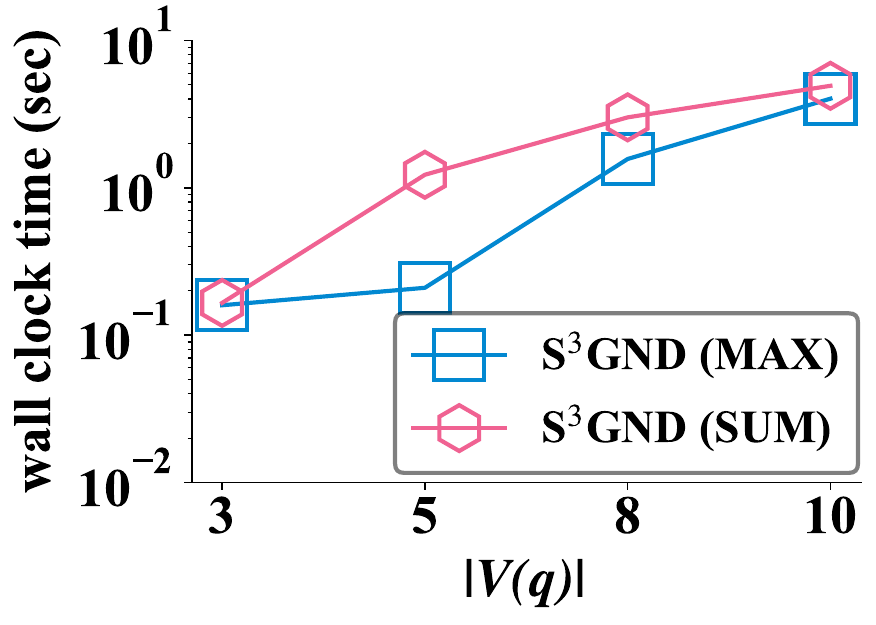}
            \label{subfig:zipf-vq}
    }
    \subfigure[{\scriptsize time vs.$|V(G)|$ ($Syn\text{-}Zipf$)}]{
            \includegraphics[width=0.15\linewidth]{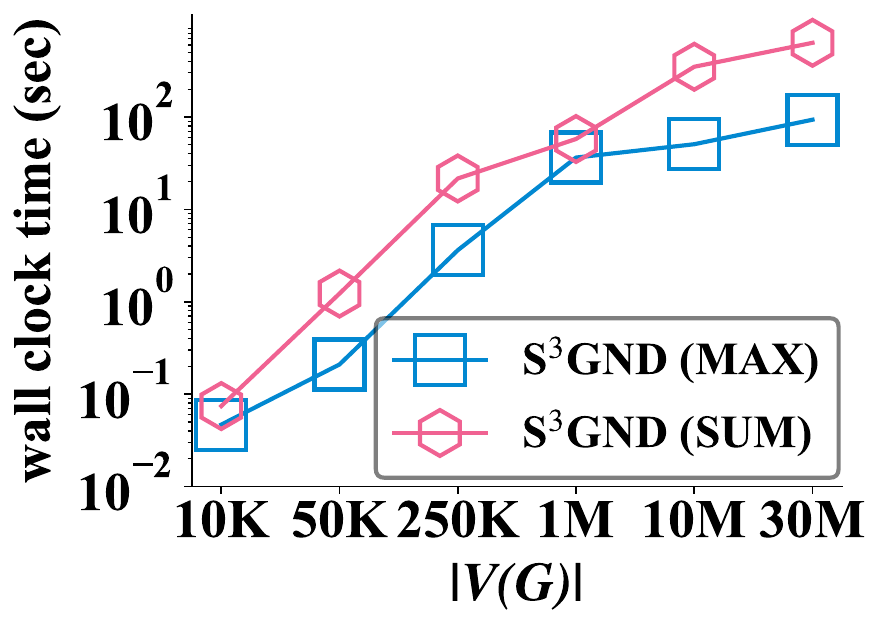}
            \label{subfig:zipf-vG}
    }
    
    \vspace{-0.1in}
    \caption{The S$^3$GND query efficiency on synthetic graphs.}
    \label{fig:parameter}
\end{figure*}

\subsection{S$^3$GND Efficiency Evaluation}


\noindent{\bf The S$^3$GND Efficiency Over Real/Synthetic Graphs:} Figure~\ref{fig:baseline} compares our S$^3$GND approach with S$^3$AND and $BF$ baselines on real-world and synthetic graphs, with default parameter settings. From experimental results, our S$^3$GND approach is 1-2 orders of magnitude faster than S$^3$AND and BF methods for both SUM and MAX aggregate functions $f(\cdot)$. $BF$ produces far more false positives (due to the hash collisions) than our S$^3$GND approach, and therefore requires significantly more time for online refinement than S$^3$GND.
Moreover, S$^3$AND cannot handle edge weights in weighted graphs under S$^3$GND semantics, therefore, cannot effectively filter out S$^3$GND false positives, compared with our S$^3$GND approach. Overall, S$^3$GND can achieve lower wall clock time than the other two baselines (e.g., for the MAX aggregate, our S$^3$GND query time is $0.014\sim 72.922~sec$ for real-world graphs and $0.156\sim 2.712$~sec for synthetic graphs). 


To further verify the robustness and efficiency of our S$^3$GND approach, we test different parameters (i.e., $\delta, |v_i.W|, |\sum|,$ $|V(q)| , |V(G)|$) over three synthetic graphs, \textit{Syn-Uni}, \textit{Syn-Gau}, and  \textit{Syn-Zipf} below.

\noindent{\bf The Efficiency w.r.t. the Threshold, $\delta_{MAX}$, of MAX Neighbor Difference:} Figures~\ref{subfig:uni-max}, ~\ref{subfig:gau-max}, and ~\ref{subfig:zipf-max} evaluate the S$^3$GND efficiency for MAX GND aggregate (i.e., $f=MAX$) on the three synthetic graphs, where the GND threshold $\delta_{MAX}$ varies from $1$ to $4$, and other parameters are set to default values.
From experimental results, we can see that as the threshold $\delta_{MAX}$ increases, the online query times for both S$^3$GND and S$^3$GND also increase. This is because a larger threshold yields more candidate vertices, which requires more refinement time. Nonetheless, the time cost of S$^3$GND still remains low (i.e., $0.12 \sim 6.84$~sec). 

\noindent{\bf The Efficiency w.r.t. Threshold, $\delta_{SUM}$, of SUM Neighbor Difference:} Figures~\ref{subfig:uni-sum},~\ref{subfig:gau-sum}, and ~\ref{subfig:zipf-sum} illustrate the S$^3$GND efficiency for SUM GND aggregate (i.e., $f=SUM$) over synthetic graphs, where the GND threshold $\delta_{SUM}$ varies from $2$ to $5$, and other parameters are by default.
In these figures, larger $\delta_{SUM}$ threshold incurs more candidate vertices for the refinement, leading to a higher query processing cost. 
Nevertheless, the wall clock time of our S$^3$GND approach remains low (i.e., $1.22 \sim 5.64$~sec).

\noindent{\bf The Efficiency w.r.t. the Number, $|v_i.W|$, of Keywords Per Vertex:} Figures~\ref{subfig:uni-w},~\ref{subfig:gau-w}, and ~\ref{subfig:zipf-w} examine the S$^3$GND query performance by varying the number, $|v_i.W|$, of keywords per vertex from $1$ to $5$ over synthetic graphs, where other parameters are set to default values. In these figures, with more keywords in $v_i.W$ of per vertex $v_i$, the pruning power of \textit{keyword embedding MBR pruning} becomes lower, which thus leads to more candidate vertices and higher time cost.
Nonetheless, the S$^3$GND query cost remains low (i.e., $0.018 \sim 8.348$~sec) for different $|v_i.W|$ values.

\noindent{\bf The Efficiency w.r.t. the Size, $|\Sigma|$, of Keyword Domain $\Sigma$:}
Figures~\ref{subfig:uni-domain}, ~\ref{subfig:gau-domain}, and ~\ref{subfig:zipf-domain} evaluate the S$^3$GND query efficiency, by setting $|\Sigma|$ to $10$, $30$, $50$, and $80$, where other parameters are set to the default values. 
From the results, with the same keyword per vertex number, a larger $|\Sigma|$ disperses keywords more widely across the domain, which enhances the efficacy of \textit{keyword embedding MBR pruning}, thereby reducing the number of candidate vertices. 
So that, for larger $|\Sigma|$ value, the S$^3$GND query cost decreases and remains low (i.e., $0.096 \sim 10.695$~sec).

\noindent{\bf The Efficiency w.r.t. the Size, $|V(q)|$, of Query Graph $q$:} Figures~\ref{subfig:uni-vq}, ~\ref{subfig:gau-vq}, and ~\ref{subfig:zipf-vq} present the effect of query graph size $|V(q)|$ on the S$^3$GND query efficiency, where $|V(q)|$ is set to $3$, $5$, $8$, and $10$, and other parameters are set to their default values. When the query graph size, $|V(q)|$, increases, more candidate vertices relative to query vertices will be retrieved and refined, resulting in higher query costs. Nevertheless, the average query time of our S$^3$GND approach remains low (i.e., $0.12 \sim 6.98$~sec).

\noindent{\bf The Efficiency w.r.t. the Size, $|V(G)|$, of Data Graph $G$:} Figures~\ref{subfig:uni-vG}, ~\ref{subfig:gau-vG}, and ~\ref{subfig:zipf-vG} test the scalability of our S$^3$GND approach for different data graph sizes $|V(G)|$ from $10K$ to $30M$, where default values are used for other parameters. For larger data graph (i.e., larger $|V(G)|$ value), the number of candidate vertices also increases, incurring higher retrieval and refinement costs. From experimental results, we can see that for a large graph, $Syn\text{-}Uni$, with as many as $30M$ vertices, the time cost is $586.364$~sec, which confirms the efficiency and scalability of our S$^3$GND approach over large-scale graphs.

Experimental results for other parameters and/or real-world graphs are similar and thus omitted here due to space limitations.

\vspace{-2ex}
\subsection{The Evaluation of S$^3$GND Offline Pre-Computations}
In this subsection, we report the offline computational cost of S$^3$GND, including the time cost of auxiliary data pre-computation and index construction, and the space cost of index storage.

\begin{figure}[t]
    \centering
    \subfigure[real-world graphs]{
        \includegraphics[width=0.48\linewidth]{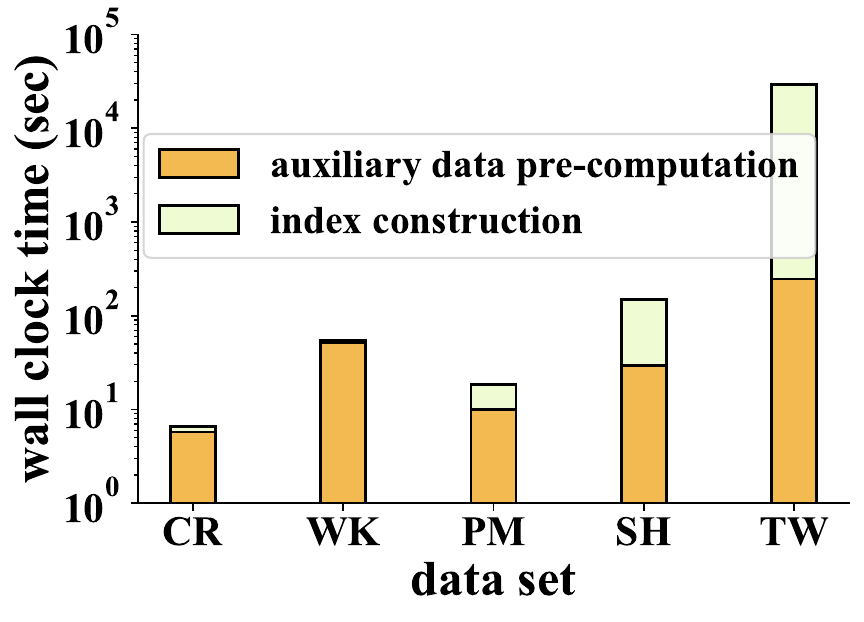}
            \label{subfig:S3GND_real_ofline-online}
    } \hspace{-0.2cm}
    \subfigure[synthetic graphs]{
        \includegraphics[width=0.48\linewidth]{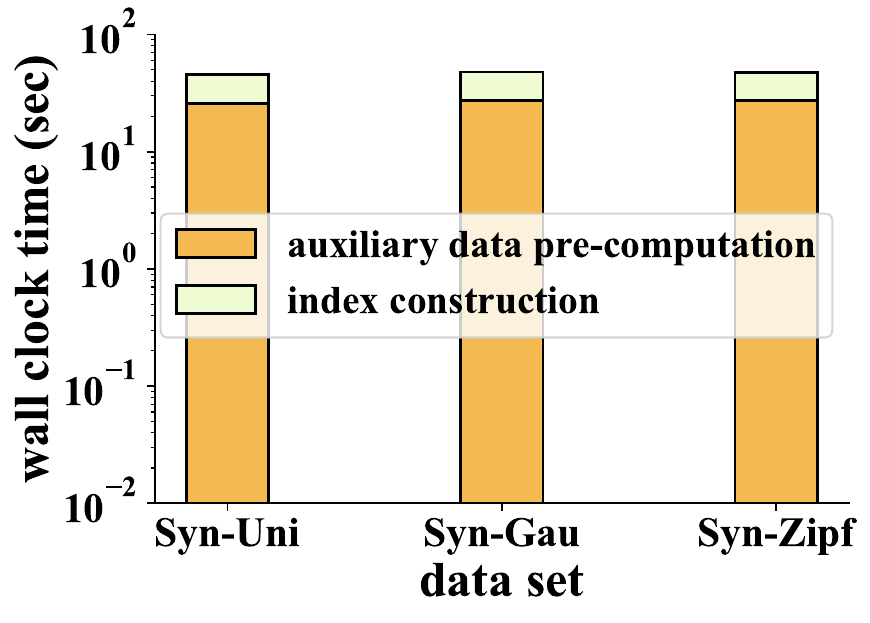}
        \label{subfig:S3GND_synthetic_ofline-online}
    }
    \vspace{-3ex}
    \caption{The time costs of offline auxiliary data pre-computation and index construction.}
    \label{fig:offline time}
\end{figure}

Figure~\ref{fig:offline time} presents the S$^3$GND offline pre-computation cost over real-world and synthetic graphs, where parameters are set to default values. From Figure~\ref{subfig:S3GND_real_ofline-online}, for the real-world graph size from $2K$ to $2.3M$, the overall offline pre-computation time (including the time costs of auxiliary data pre-computation and index construction ) varies from $6.55~sec$ to $8.2~h$. 
For synthetic graphs (with the graph size $|V(G)|=50K$) in Figure~\ref{subfig:S3GND_synthetic_ofline-online}, the overall offline pre-computation time varies from $25.78~sec$ to $27.32~sec$.


Figure~\ref{fig:space} illustrates the space costs of the transformed keyword hypergraph $H$, the pre-computed index $\mathcal{I}$, and the raw data graph $G$ over both real-world and synthetic graphs. 
Figures~\ref{subfig:S3GND_index_space_real} and ~\ref{subfig:S3GND_index_space_syn} show that the hypergraph $H$ requires between 3X and 43X less storage space than graph $G$, while the index $\mathcal{I}$ needs between 2X and 20X less storage space than graph $G$.

Overall, our experimental results indicate that the offline pre-computations of our S$^3$GND approach can effectively support large-scale graphs with reasonable time and space costs.


%

\begin{figure}[t]
    \centering
    \subfigure[real-world graphs]{
            \includegraphics[width=0.49\linewidth]{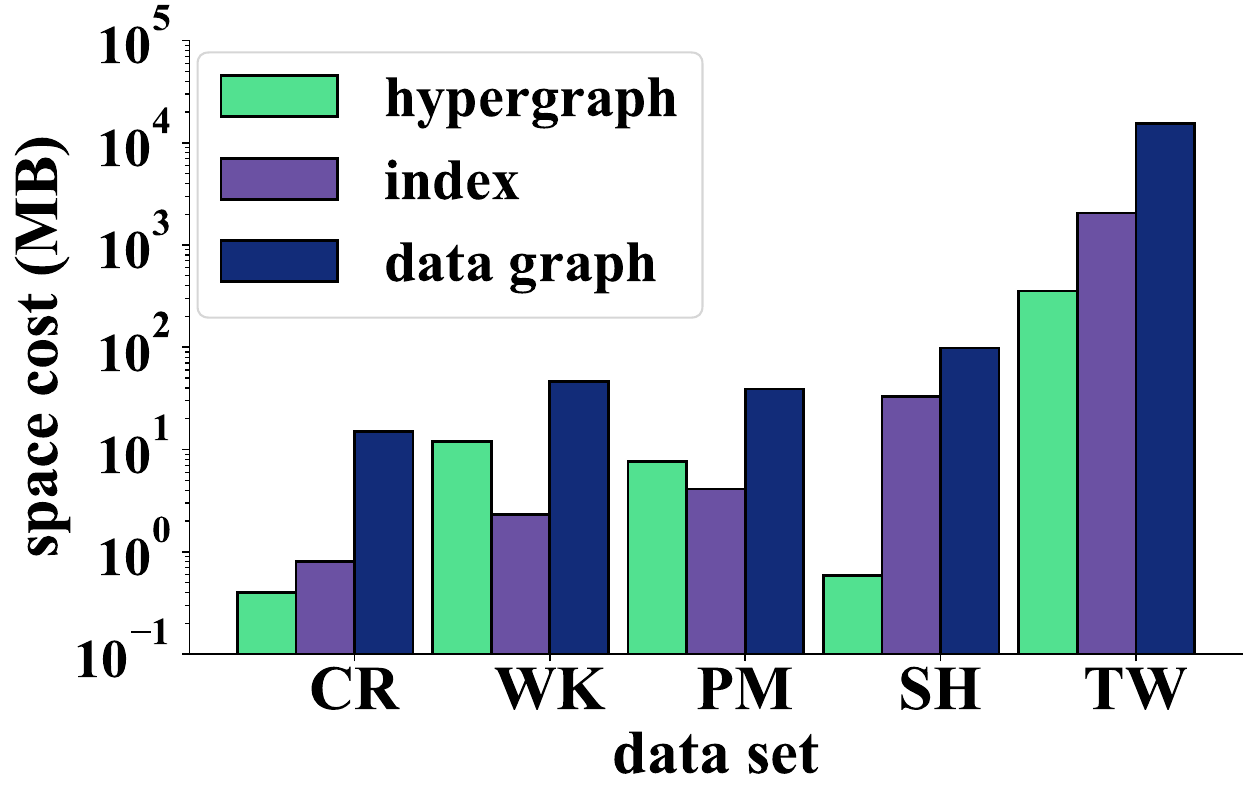}
            \label{subfig:S3GND_index_space_real}
    } \hspace{-0.4cm}
    \subfigure[synthetic graphs]{
        \includegraphics[width=0.49\linewidth]{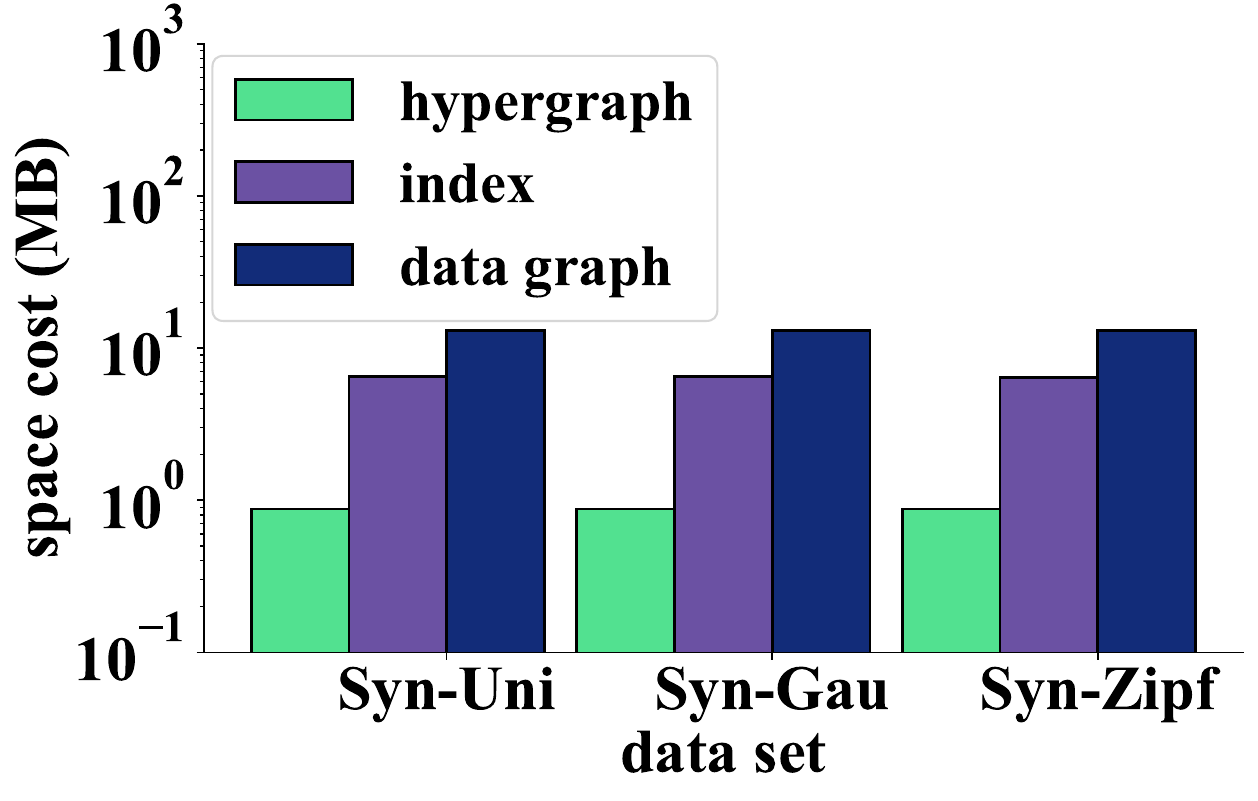}
        \label{subfig:S3GND_index_space_syn}
    }\vspace{-2ex}
    \caption{The space cost of the offline pre-computed hypergraph and index.}
    \label{fig:space}
\end{figure}

\subsection{Case Study}
\label{sec:case}
To further validate the effectiveness of our S$^3$GND approach, we present a case study of our S$^3$GND query over a real-world graph, $shanghai$, compared with S$^3$AND. Given a same query graph $q$ ($|V(q)|=5$) and the aggregate function $f=MAX$, Figure~\ref{fig:case} evaluates the differences between top-1 query result of S$^3$GND (denoted as $g_{GND}$) and that of S$^3$AND~\cite{wen2025s3and} (denoted as $g_{AND}$), under GND and AND~\cite{wen2025s3and} semantics, respectively. 

From  Figure~\ref{fig:case}, the same colors of vertices in subgraphs ($g_{GND}$ or $g_{AND}$) and query graph $q$  indicate the matching of vertex keyword sets (i.e., the keyword set containment relationship). Under different semantics, the AND score of the S$^3$AND answer $g_{AND}$ is 2, which is the same as that of our S$^3$GND answer $g_{GND}$. However, compared with S$^3$AND, S$^3$GND can retrieve more edges with large weights (e.g., edge $e_{v_{66371},v_{68338}}$ in $g_{GND}$ has weight 5), which are usually more important (e.g., higher connectivity) than those with lower edge weights. 
Moreover, under GND semantics, the S$^3$AND result $g_{AND}$ has a GND score of $GND(q, g_{AND})=5$, whereas our S$^3$GND answer $g_{GND}$ has a much smaller (better) GND score $GND(q, g_{GND})=2$, which confirms the quality of our retrieved S$^3$GND answer over the weighted graph.




\nop{\color{orange}
Under the keyword-set matching condition, this case intuitively demonstrates the difference in matching accuracy between the S$^3$GND and S$^3$AND algorithms.
Although the candidate subgraph returned by S$^3$AND achieves a topological match (i.e., $AND(q, g_{1}^{AND})=0$), it shows a significant deviation from the query graph in terms of semantic information such as edge weights.
For example, in the candidate subgraph $g_{1}^{AND}$, the weight of edge $e_{v_1,v2}$ cannot be precisely aligned with the weight of $e_{q_1,q_2}$ in the query graph $q$, resulting in a lower GND score under different aggregate function $f$ (i.e., $GND(q, g_{1}^{AND})=0.25$ if $f=MAX$, and $GND(q, g_{1}^{AND})=0.5$ if $f=SUM$).
In contrast, our S$^3$GND algorithm guarantees not only structural matching but also precise semantic matching. 
The edge weights of the returned subgraph perfectly match the GND semantic matching of the query graph.
This is because while S$^3$AND focuses solely on topological connectivity, S$^3$GND considers both topological structure and content semantics. As a result, S$^3$GND enables the retrieval of higher-quality subgraph results that better align with practical needs.
}

\begin{figure}[!t]
    \centering
    \includegraphics[width=0.995\linewidth]{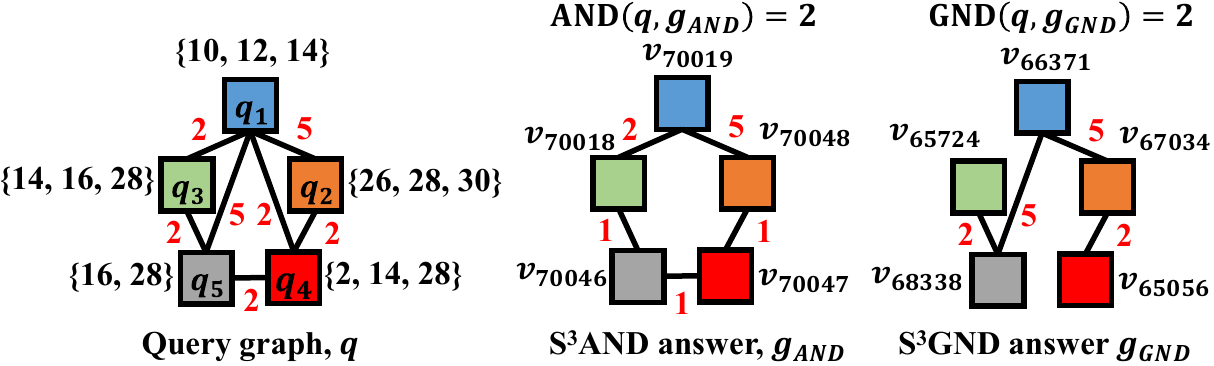}\vspace{-1ex}
    \caption{A case study of S$^3$GND and S$^3$AND queries over $shanghai$ graph with query graph size $|V(q)|=5$ and $f=MAX$.}
    \label{fig:case}
\end{figure}

\balance
\section{Related Work}
\label{sec:related_work}

\noindent{\bf Graph Data Management:} Prior works on graph data management have primarily focused on effective storage~\cite{li2023graphar, hao2025rapidstore}, indexing~\cite{yang2024revisiting,sehgal2025navix}, and querying over massive graph structures~\cite{zhang2024comprehensive,gao2025efficient}. Many existing graph management systems~\cite{chakraborty2025robust, skavantzos2025entity} provide persistent storage services and graph analysis, enhanced by effective graph indexes~\cite{zhang2025botbin, mang2024efficient}. Various fundamental graph queries have been studied, including path search queries~\cite{liu2025wolverine, shen2025efficient}, keyword search queries~\cite{lin2015reverse, jiang2023dkws}, (sub)graph matching~\cite{zhang2024comprehensive, ye2024efficient}, community search/detection~\cite{hsu2025attacking, desilva2026eiga}, and so on. In this paper, we focus on the subgraph similarity search problem over a large-scale weighted graph, in which each vertex is associated with a set of keywords. 




\noindent{\bf Subgraph Similarity Search:} The subgraph similarity search has attracted substantial attention for decades~\cite{yuan2015graph, liu2025masked, wen2025s3and, yuan2012efficient, yan2005substructure, yuan2015efficient, chang2022accelerating}.
Existing approaches to subgraph similarity search can be broadly grouped into \textit{exact methods based on filtering and pruning}~\cite{cheng2007fg, yan2005graph, han2013turboiso, kim2022efficient, kim2021boosting} and \textit{approximate methods based on partitioning and heuristics}~\cite{zhao2013partition, zhao2017efficient}.
The FG-index~\cite{cheng2007fg} is an exact method that builds a hierarchical graph index and integrates pruning strategies to eliminate false positives, thereby supporting exact search for subgraphs similar to a given query graph. For approximate methods, Zhao et al.~\cite{zhao2017efficient} employed graph partitioning techniques together with heuristic search to retrieve the top-$k$ similarity results for large-scale networks. 
While many existing studies incorporate graph topology and/or node attributes into graph similarity measures, they typically assume an unweighted graph model, without considering edge weights (and their similarities) or the fact that each vertex is associated with a single keyword/label. In contrast, in this work, our S$^3$GND problem takes into account a more general weighted graph model, containing edge weights and a keyword set per vertex, and, moreover, provides a novel subgraph matching semantic, S$^3$GND, over such a general graph, whose graph similarity measure integrates edge weights, topology, and vertex keyword sets.

Recently, learning-based subgraph similarity search~\cite{ling2021multilevel, ye2024efficient, bai2021glsearch, wang2023neural, yang2023extract, sima2025deep, nikolentzos2017matching, bai2019simgnn} has been introduced to better balance query efficiency and matching accuracy. 
Ye et al.~\cite{ye2024efficient} proposed an effective and efficient GNN-PE approach, which achieves exact subgraph search via the learned dominance-embeddings for paths. Moreover, some works on approximate subgraph similarity~\cite{nikolentzos2017matching, bai2019simgnn} measured subgraph similarity based on node embeddings and distances in the embedding space. These prior works usually considered either a simple graph model with a single keyword per vertex and unweighted edges, or a graph matching measure (e.g., graph isomorphism, graph edit distance, etc.), which have different graph models and/or matching semantics from our S$^3$GND problem settings (e.g., weighted graph with a keyword set per vertex, and S$^3$GND semantic that considers edge weights/keyword sets/topology). Thus, we cannot directly apply previous learning-based techniques for subgraph similarity search to processing our S$^3$GND queries. 

\nop{
With the continuous expansion of graph data, an increasing number of learning-based subgraph similarity search methods~\cite{ling2021multilevel, ye2024efficient, bai2021glsearch, wang2023neural, yang2023extract, sima2025deep, nikolentzos2017matching, bai2019simgnn} have been introduced to better balance query efficiency and matching accuracy. 
Ye et al.~\cite{ye2024efficient} propose the GNN-PE approach, which achieves precise subgraph search by learning path-dominance embeddings for subgraphs. Nevertheless, GNN-PE is constrained by its dependence on a single keyword representation and its substantial pre-computation overhead.
The methods in~\cite{nikolentzos2017matching, bai2019simgnn} measure subgraph similarity based on node embeddings and distances in the embedding space, thus producing approximate subgraph similarity results.
While these approaches can return answers quickly, the matched subgraphs often fail to exhibit consistent similarity in both structure and content, and they typically require a large amount of training data to achieve satisfactory accuracy.
In contrast, our s$^3$GND method does not embed the original nodes directly; instead, it embeds the associated keywords. We leverage Minimum Bounding Rectangle (MBR) containment relationships among node keyword sets for both embedding training and containment decisions, thereby substantially reducing the need for training data generation while ensuring accurate query results.
}

\section{Conclusion}
\label{sec:conclusion}

In this paper, we formulate a novel problem, \textit{subgraph similarity search under the generalized neighbor difference semantics} (S$^3$GND), which has broad applications (e.g., protein discovery, social network analysis, and recommender systems) in real-world scenarios.
To efficiently and accurately address the S$^3$GND problem, we propose an effective learning-based approach for training a keyword-hypergraph-based embedding model and generating high-quality embedding MBRs for vertex keyword sets, which can be used for reducing the S$^3$GND search space. We devise a tree index over offline pre-computed auxiliary data (including keyword embedding MBRs and sorted lists of edge weights), which can facilitate the proposed effective pruning strategies (i.e., \textit{keyword embedding MBR}, \textit{ND lower bound}, and \textit{GND lower bound pruning}) for our developed efficient online S$^3$GND query algorithm. Finally, we conduct extensive experiments on real/synthetic graphs to validate the effectiveness and efficiency of our S$^3$GND approach.

\clearpage

\bibliographystyle{ACM-Reference-Format}

\bibliography{dblp}

\end{document}